\documentclass[fleqn,10pt]{wlscirep}
\usepackage[utf8]{inputenc}
\usepackage[T1]{fontenc}
\usepackage{amsmath,amsfonts,amsthm}
\usepackage{array}
\usepackage{makecell}
\usepackage{multirow}
\usepackage{booktabs}
\usepackage[flushleft]{threeparttable}

\theoremstyle{plain}
\newtheorem{lem}{Lemma}
\newtheorem{thm}{Theorem}

\theoremstyle{definition} 
\newtheorem{defn}{Definition}

\newcommand{\Z}{\mathbb{Z}}
\newcommand{\R}{\mathbb{R}}

\renewcommand{\P}{\mathbb{P}}

\newcommand{\sgn}{\mathrm{sgn}}

\newcommand{\D}{\mathbb{D}}

\renewcommand{\O}{\mathbb{O}}

\newcommand{\supplementarysection}{%
  \setcounter{figure}{0}
  \let\oldthefigure\thefigure
  \renewcommand{\thefigure}{S\oldthefigure}
  \let\oldchapter\chapter
  \renewcommand{\chapter}{
    \let\thefigure\oldthefigure
    \let\chapter\oldchapter
    \oldchapter
  }
}

\usepackage{mathtools}

\usepackage[linesnumbered,ruled,vlined]{algorithm2e}

\title{Privacy Guarantees for Personal Mobility Data in Humanitarian Response}

\author[1,+]{Nitin Kohli}
\author[2,+]{Emily Aiken}
\author[1,2,*]{Joshua Blumenstock}
\affil[1]{Center for Effective Global Action, UC Berkeley, 94704 USA}
\affil[2]{School of Information, UC Berkeley, 94704 USA}

\affil[+]{These authors contributed equally to this work.}
\affil[*]{jblumenstock@berkeley.edu}

\begin{abstract}

 Personal mobility data from mobile phones and other sensors are increasingly used to inform policymaking during pandemics, natural disasters, and other humanitarian crises. However, even aggregated mobility traces can reveal private information about individual movements to potentially malicious actors. This paper develops and tests an approach for releasing private mobility data, which provides formal guarantees over the privacy of the underlying subjects. Specifically, we (1) introduce an algorithm for constructing differentially private mobility matrices and derive privacy and accuracy bounds on this algorithm; (2) use real-world data from mobile phone operators in Afghanistan and Rwanda to show how this algorithm can enable the use of private mobility data in two high-stakes policy decisions:  pandemic response and the distribution of humanitarian aid; and  (3) discuss practical decisions that need to be made when implementing this approach, such as how to optimally balance privacy and accuracy. Taken together, these results can help enable the responsible use of private mobility data in humanitarian response.

\end{abstract}

\begin{document}

\flushbottom
\maketitle

\thispagestyle{empty}

\section*{Introduction}
\label{sec:introduction}
Personal mobility data has the potential to provide critical information to guide the response to humanitarian crises, and to advance several of the Sustainable Development Goals (SDGs). For instance, recent work has shown that mobility data can be used to model and prevent the spread of epidemics \cite{wesolowski2015, wesolowski2012, bengtsson2015using}, monitor and assist displaced populations after natural disasters \cite{lu2012, pastorescuredo2014,tai_mobile_2022}, target emergency cash transfers \cite{aiken_machine_2021,blumenstock_machine_2020}, and flag riots and violence against civilians \cite{dobra2015, gundogdu2016}. The COVID-19 pandemic in particular demonstrated that aggregated personal mobility data can provide critical insight into population movement \cite{kraemer2020, warren2020, chinazzi2020,ilin_public_2021, lai2020effect}, and resulted in a push for collaborations among governments, technology companies, and researchers to leverage these data to inform evidence-based pandemic response measures \cite{oliver2020}.

At the same time, personal mobility data is, by definition, personal, so its analysis -- even for humanitarian purposes -- may create privacy risk. Individual mobility data can expose personally identifying and sensitive information, including individuals' home and work locations, travel patterns, and interpersonal interactions. Moreover, since certain locations are endowed with social meaning, mobility data can be used to draw inferences about an individual’s life, including their political preferences\cite{thompson2019twelve} and sexual orientation\cite{ovide2021, boorstein2021}. Even coarsening the spatial and temporal characteristics of individual locations traces provides little privacy protection, due to the idiosyncratic nature of human movements throughout daily life\cite{demontjoye2013}. There has thus been widespread criticism and controversy about the use of such data by governments \cite{pclobreport2020}, private companies \cite{thompson2019twelve}, civil society \cite{hosein2013aiding}, and researchers \cite{taylor2016}.

To mitigate such personal inferences about individuals from digitally-derived  mobility data -- and the potential harms that could follow from these revelations -- active measures must be taken to ensure that individual privacy is protected when such data are analyzed or released. The traditional approach to protecting privacy is to release only aggregated statistics about group movement rather than individual mobility trajectories \cite{demontjoye2018}. However, such aggregation is typically insufficient to de-identify individuals in mobility traces \cite{demontjoye2013, xu}, so other papers have turned to other techniques including distorting  spatiotemporal information \cite{abul2008, primault2015}, obfuscating individual trajectories \cite{meyerowitz}, and data slicing and mixing \cite{shi2010}.

This paper develops and tests an approach to producing \textit{provably private} mobility statistics from personal data based on differential privacy \cite{dwork2006calibrating}. A number of previous studies have shown how differentially private approaches can add carefully calibrated noise to aggregated mobility statistics to ensure that the amount of individual information leaked by such statistics is bounded\cite{mir2013dp, qardaji2013differentially, aktay2020google, machanavajjhala2008privacy, pratesi2014privacy, jiang2013publishing, ghane2019tgm, shao2013publishing, li2017achieving, savi2023standardised}. In comparison to past work, the focus of our study is on understanding the implications and tradeoffs that arise when using private mobility statistics in humanitarian settings. We are specifically interested in how differential privacy impacts \textit{intervention accuracy}, i.e., the accuracy (and effectiveness) of policy decisions that are based on private mobility data rather than standard (non-private) data.

This paper makes three main contributions. First, we develop an algorithm for computing differentially private mobility information, and derive the guarantees that this algorithm provides for both data privacy and intervention accuracy. Second, we provide experimental results on the tradeoff between privacy and the accuracy of downstream policy decision in two real-world humanitarian contexts: response to (1) pandemics and (2) natural disasters and violent events. Third, we discuss the nuanced implementation choices that policymakers and algorithm developers are likely to grapple with in deploying such systems, and their effects on the privacy-accuracy tradeoff and ability to effectively deliver humanitarian aid.

\section*{Results}
\label{sec:results}

\subsection*{Building and Testing a Differentially Private Mobility Matrix}
\label{sec:matrices}

Our first set of results formalize the notion of a \textit{private mobility matrix}, propose a novel algorithm for generating such matrices, and derive guarantees on the privacy and accuracy provided by this algorithm. In general, a \textit{mobility matrix} is a data structure that quantifies the extent to which a population moves between regions in a fixed period of time -- see Figure \ref{fig:matrix}A for an example of a an O-D matrix derived from mobile phone records. The most common form of mobility matrix, an origin-destination matrix (hereafter abbreviated as an O-D matrix), indicates for each pair of regions $(i,j)$ the number of trips taken by individuals from $i$ to $j$. Mobility matrices are widely used by scholars and policymakers, and are playing an increasingly critical role in many humanitarian settings \cite{chinazzi2020, ilin_public_2021, wesolowski2012, wesolowski2015}. However, as discussed above, the raw data in mobility matrices are sensitive, and can compromise the privacy of the individuals whose data are used to construct those matrices \cite{shaham2022differentially, machanavajjhala2008privacy}. 

Algorithm~\ref{algo:od}, described in detail in  the ``Methods'' section, provides a computationally efficient method to construct \textit{differentially private}\cite{dwork2006calibrating} origin-destination matrices from individual mobility data. At its core, this algorithm adds very carefully calibrated noise to the underlying data, in a specific way that both protects the privacy of individuals and ensures the fidelity of the resulting O-D matrix -- see Figure \ref{fig:matrix} for a schematic diagram of this framework. In the \textit{Supplementary Information}, we  prove that our algorithm satisfies $\epsilon$-differential privacy (Theorem 1, \textit{Supplementary Information}). This proof mathematically guarantees that no one with access to the  private O-D matrix could infer ``too much'' about any individual trip (or any individual) whose data was contained in the matrix. The algorithm quantifies ``too much'' using a privacy loss parameter $\epsilon \ge 0$ (Definition 1, \textit{Supplementary Information}).
This privacy loss parameter $\epsilon$ is thus critical to understanding the function of the algorithm, as it governs the maximum information about an individual trip (or an individual person) is leaked by a private mobility matrix.  Smaller values of $\epsilon$ provide more robust privacy guarantees while larger values increase the risk that individual data could be compromised.

However, a central concern --- and a focal point of our work --- is that increases in privacy come at a cost. In particular, as $\epsilon$ decreases, the resulting private mobility matrix looks less and less like the true (non-private) mobility matrix.  In the \textit{Supplementary Information}, we derive formal guarantees of the worst-case accuracy loss of the private O-D matrix algorithm as a function of $\epsilon$. In particular, Theorems 2, 3, and 4 show that the accuracy loss of our algorithm exponentially decreases as $\epsilon$ increases; Theorem 5 further shows that, when using multiple private mobility matrices to estimate changes in population flows, the detection error decreases according to $\Theta(\epsilon\exp(-\epsilon))$ as $\epsilon$ increases. 

To provide intuition for the tradeoff between privacy and accuracy, Figure \ref{fig:boxplots} illustrates how the privacy loss parameter $\epsilon$ relates to differences between the original O-D matrix and the private O-D matrix. To construct the figure, we use a large dataset of mobile phone records provided by one of the largest mobile network providers in Afghanistan, covering roughly 7 million individuals and containing over 3 billion phone transactions over 305 days (for details about the data used in this paper, see Table \ref{table:data}). We use our algorithm to construct a private O-D matrix from these data, and compare this private matrix to a non-private O-D matrix constructed from the raw data.
In the figure, we observe that the errors introduced by privatization are quite small. For $\epsilon=0.1$ the median absolute error (i.e., the median difference between any given private matrix count and the corresponding non-private matrix count) is 6 trips at the admin-2, and the median relative error (the absolute error divided by the non-private matrix count) is 2.54\% at the admin-2 level. For $\epsilon=1$ --- a relatively weaker level of privacy protection --- the median absolute error is 0 trips and the median relative error is under 0.01\% at the admin-2. As a general matter, the relative accuracy of our algorithm tends to increase as the difference between a non-private count and the suppression threshold increases; that is, lower mobility areas tend to have lower relative accuracy than higher mobility areas after applying our algorithm. This observation is consistent with prior research that the accuracy of differentially private analyses (in general) tend to increase as the size of the underlying data increases \cite{dwork2019differential}.

\subsection*{How does privatizing mobility data affect humanitarian interventions based on those data?}
\label{sec:accuracy}

The preceding results suggest that our algorithm introduces only modest errors when creating a private version of the O-D matrix. However, it is not clear if and how those modest decreases will impact downstream policy decisions based on private O-D matrices. Our next set of results therefore illustrate how policies in two key humanitarian settings are affected by the privatization of mobility data.

Specifically, we evaluate private O-D matrices in the context of a hypothetical pandemic scenario in Afghanistan, and in natural disaster scenarios in Afghanistan and Rwanda. These case studies leverage location data recorded in call detail records (CDR) collected by mobile phone operators (see Table~\ref{table:data}) --- the same sort of data that has raised privacy concerns in a wide range of contexts discussed in the Introduction \cite{demontjoye2013, pclobreport2020, hosein2013aiding, taylor2016}. The raw CDR indicate, for each phone call made on the operator's network, the date, time, and duration of the call, identifiers for each subscriber involved in each call, and the location of the cell tower through which each call was placed. This last piece of information is what makes it possible to infer the approximate location of each originating subscriber at the time of the phone call. From these raw location trace data, we calculate daily O-D matrices recording daily movements between each pair of administrative subdivisions in each country (see Figure~\ref{fig:matrix}C for an example), and then show how those matrices -- and the private version of them -- influence downstream humanitarian interventions. 

\subsubsection*{Non-pharmaceutical interventions in pandemic response}
Our first case study compares the effectiveness of private and non-private mobility data in guiding 
public health interventions intended to stop the spread of contagious disease. This is an increasingly common application of mobility data, and an approach that was used to inform non-pharmaceutical interventions to the COVID-19 pandemic \cite{kraemer2020, warren2020, chinazzi2020,ilin_public_2021}. Our analysis simulates a pandemic scenario, where the spread of the disease is modeled using a classic SIR (Susceptible, Infected, Recovered) model\cite{kermack1927contribution}, adapted to account for the fact that many modern diseases are influenced by human mobility\cite{goel2020mobility}. We consider policy decisions that may be enacted when disease prevalence estimates from the mobility-informed SIR model exceed a certain threshold (such as travel restrictions, medication distribution, and so on), and test the extent to which the timing of such decisions differs when private and non-private O-D matrices are used for mobility modeling. 

To conduct these simulations, we use the same large dataset of mobile phone records from Afghanistan in 2020 used to calculate errors in matrix counts from privatization (Figure \ref{fig:boxplots}). Daily O-D matrices at the province (admin-2) and district (admin-3) level derived from the CDR for 305 days are used as inputs to the mobility-informed SIR model, assuming that mobility data is available on a daily basis (but epidemiological data is available only at the start of the pandemic). We simulate three possible pandemics --- pandemics initiating in Hirat, in Kabul, and in a randomly selected origin region --- with 1\% of individuals in the original region infected on the first day.  In these simulations, we measure intervention accuracy by studying the  daily (binary) decisions policymakers would make about whether or not to institute local anti-contagion policies, comparing the decisions that would be made with private data to those that would be made with non-private data. We measure the accuracy, precision, and recall of these binary decisions, with the assumption that policies are put into place at a threshold of 20\% estimated disease prevalence.

As shown in Table \ref{table:pandemic}, we find that intervention accuracy remains fairly high when using differentially private O-D matrices at the admin-2 (province) level (e.g. accuracy of policy decisions = 96-99\% for $\epsilon=0.5$). However, when simulations are conducted with a finer spatial granularity --- modeling disease spread at the district, or admin-3, level --- intervention accuracy drops, and there is higher accuracy variance across pandemic simulations (e.g. accuracy of policy decisions = 72-91\% for $\epsilon=0.5$). This difference highlights the importance of implementation choices like spatial granularity in deployments of the private O-D matrix algorithm: in general, the accuracy loss incurred by privatization will be larger at high spatial resolutions, when  individual entries in the O-D matrix counts are smaller. 

The value of the privacy loss parameter $\epsilon$ also has important implications for intervention accuracy. Higher values of $\epsilon$ --- which correspond to less noise being added to mobility matrices in the private O-D matrix algorithm --- lead to higher intervention accuracy (e.g. accuracy of policy decisions = 97-99\% at the admin-2 level and 77-91\% at the admin-3 level for $\epsilon=1.0$). Lower values of $\epsilon$ --- meaning more noise added to mobility matrices --- lead to lower intervention accuracy (e.g. accuracy of policy decisions = 90-98\% at the admin-2 level and 64-83\% at the admin-3 level for $\epsilon=0.1$). The choice of $\epsilon$ ultimately comes down to a policy decision on the relative importance of privacy and accuracy; we discuss methods for tuning the privacy-accuracy tradeoff at the end of the results section.

We also note that the precision and recall of policy decisions based on private O-D matrices can depend on the source of the pandemic, even for the same value of $\epsilon$. For instance, compared to Kabul, which has the highest precision and recall (Table~\ref{table:pandemic}), performance drops when the pandemic originates in Hirat (or in a random location). This occurs because the volume of trips is much lower in less populous regions: whereas 18,602 trips originate from Kabul on an average day, only 3,733 originate in Hirat. Thus, when we apply our algorithm, the noise and suppression introduced into small counts has a larger effect on the mobility-adapted SIR dynamics than for larger counts.  This sensitivity highlights the importance of empirically evaluating the privacy-accuracy tradeoff in context, before deciding if and how to implement in policy.


\subsubsection*{Geographic targeting of humanitarian aid after natural disasters and violent events}

Our second case study illustrates how the privatization of mobility data can impact the allocation of humanitarian aid (such as medical resources, temporary shelter, and food) following natural disasters and other shocks. This application is motivated by the fact that mobility data derived from mobile phones and social media have been used to inform post-disaster humanitarian interventions, including following the 2010 earthquake in Haiti \cite{lu2012}, the 2019 wildfires in California \cite{kikade}, and most recently during the 2023 earthquake in Turkey \cite{crisisready}. In our hypothetical aid-targeting scenario we simulate a short-term shock --- such as an earthquake, landslide, flood, or violent event --- and measure the degree of out-migration (as the total out-trips observed in an O-D matrix) in the following week. We assume a government or other policy actor seeks to provide humanitarian aid to the migrants in the areas with the most out-migration.

For our aid targeting simulations, we use two CDR datasets related to real-world natural disasters and violent events in simulating the implications of our private O-D matrix algorithm for targeting humanitarian aid. In Afghanistan, we study the battle of Kunduz in 2015, which resulted in the displacement of more than 100,000 people according to official statistics\cite{militarytimes2015}. We focus on displacement in the week following the largest single day of the battle, using a CDR dataset of around 64 million transactions associated with 2.8 million mobile subscribers (Table \ref{table:data}). In Rwanda, we study mobility in the week following the Lake Kivu Earthquake of 2008, using a CDR dataset that includes around 13 million transactions associated with half a million mobile subscribers (Table \ref{table:data}). In these settings we measure intervention accuracy using two metrics: first, we compare estimated total out-migration from the affected area using private data to non-private data, and second we calculate the accuracy of the top 3 regions of out-migration (i.e., the regions receiving the most migrants after a disaster) identified in private data in comparison to those identified in non-private data. 

In Table \ref{table:targeting} we show that in both the Afghanistan and Rwanda contexts, private O-D matrices with trip-level protection have relatively low error in determining total out-migration following natural disasters and violent events in comparison to non-private data (e.g. percent error = 2.54-8.27\% for $\epsilon=0.5$). All private counts are biased downwards relative to non-private counts due to suppression of small values in private O-D matrices (Figure S2, \textit{Supplementary Information}). As in the simulations of non-pharmaceutical interventions in pandemic response, we observe that intervention accuracy using private mobility matrices is higher at the admin-2 level (percent error = 2.54-7.38\% for $\epsilon=0.5$) than at the admin-3 level (percent error = 4.98-8.27\% for $\epsilon=0.5$). Unintuitively, however, we actually observe \textit{higher} intervention accuracy for small values of $\epsilon$ (which correspond to large amounts of noise added to mobility matrices) than for large values of $\epsilon$ (which correspond to adding small amounts of noise). This arises from the use of a suppression threshold. By Theorem 2 (\textit{Supplementary Information}), a non-private count below the threshold is more likely to become non-suppressed in Step 3 of the algorithm when a smaller $\epsilon$ is used compared to a larger one, thereby making the total out-migration flow more accurate for $\epsilon = 0.1$ compared to $\epsilon \in \{0.5, 1\}$.

Private O-D matrices with trip-level protection also have relatively low error in identifying the areas of largest out-migration for targeting humanitarian aid. We focus in our simulations on selecting the top 3 geographies with the most out-migration for targeting post-disaster resources. We find that selection with private data is 90-100\% accurate for $\epsilon=0.5$ (Table \ref{table:targeting}). Unlike in measuring total out-migration, we observe little sensitivity to the spatial granularity of simulation or the value of $\epsilon$ in identifying the top regions of out-migration. This result may be at least partially explained by the fact that identifying the top regions of out-migration from a relatively small number of regions with large out-migration flows is a less nuanced task than identifying the exact degree of outflow. In general we also do not observe much sensitivity to the number of regions targeted  for humanitarian aid (for example, if 5 or 10 regions are identified instead of 3, as shown in Figure S3 (\textit{Supplementary Information}).

Private O-D matrices with individual-level protection exhibit similar behavior as those with trip-level protection (as specified by the parameter $T$ in our algorithm). Table S3 (\textit{Supplementary Information}) shows that as increasingly restrictive values of $T$ are used, our total out-migration measure under-counts total out-migration. However, the value of $T$ does not have much impact on the accuracy of identifying the top-$k$ districts of out-migration districts. Compared to trip-level privacy --- which provides weaker privacy protections to individuals --- we find that the privacy-accuracy tradeoff between individual- and trip-level privacy to be nuanced. In some situations, the accuracy of out-migration using individual-level protection is higher than using trip-level privacy (e.g., for the Lake Kivu earthquake at the admin-2 level, matrices with individual-level protection experience 4.87-5\% error (Table S3, Panel D), compared to 6.23-8.88\% error at the trip-level (Table \ref{table:targeting}, Panel B)); in other situations, the relationship is reversed (e.g., for the Battle of Kunduz at the admin-2 level, matrices with individual-level protection incur 15.39-18.15\% error (Table S3, Panel A), compared to 1.98-2.61\% error at the trip-level (Table \ref{table:targeting}, Panel A)). We observe a similar phenomenon in identifying the top-$k$ districts: in some situations, private O-D matrices with individual-level protection are more accurate at identifying the top-$k$ districts of out-migration (e.g., for the Lake Kivu earthquake at the admin-2 level, matrices with individual-level protection are 100\% accurate (Table S3, Panels D-F), whereas matrices with trip-level protection are 90.47-95.25\% accurate (Table \ref{table:targeting}, Panel B)); in other circumstances, private matrices with trip-level protection are more accurate (e.g., for the Battle of Kunduz at the admin-2 level, matrices with individual-level protection are 85.71\% accurate (Table S3, Panels A-C), whereas matrices with trip-level protection are 90.47\% accurate (Table \ref{table:targeting}, Panel A)). 

\subsection*{How should policymakers navigate the tradeoff between privacy and accuracy?}
\label{sec:tradeoff}
Our main results show that policymakers can switch from non-private to private mobility matrices --- which guarantee a certain degree of protection of individual mobility traces --- with relative little accuracy loss. For example, in simulating the accuracy of nonpharmaceutical interventions in mobility-informed pandemic response, we found an accuracy of 72-99\% for policies enacted using private mobility matrices. Accuracy is similar for policy decisions related to mobility-informed targeting of humanitarian resources following natural disasters, at 77-99\%. However, our results also underscore that there is a key tradeoff between privacy and accuracy (shown in Figure \ref{fig:boxplots} and made explicit in Theorems 1-5, \textit{Supplementary Information}). In the case of the private O-D matrix algorithm, this tradeoff is quantified by the privacy loss parameter $\epsilon$. 

What is the benefit of increasing privacy by reducing $\epsilon$? So far, we have shown that decreasing $\epsilon$ generally comes with a decrease in accuracy for downstream policy interventions. All else being equal, we would like to maintain the highest standard of accuracy in policy interventions informed by mobility data --- but doing so comes at a loss to the privacy of individual mobility traces. In this section, we use properties of differential privacy (Lemmas 2 and 3, \textit{Supplementary Information}) to derive worst-case privacy guarantees for the policy intervention simulations from the previous section. We then compare these worst-case privacy guarantees to the expected individual-level privacy loss, and discuss practical methods for policymakers to determine the appropriate tradeoff between privacy and accuracy. 

\subsubsection*{Worst Case vs. Expected Privacy Loss}
The privacy guarantees for our differentially private O-D matrix algorithm (Lemmas 2 and 3, \textit{Supplementary Information}) suggest that worst-case privacy loss for individuals can be quite high. Privacy loss is a unitless quantity that measures the maximum additional information an adversary learns about any individual when their data is used in our private O-D matrices — such information includes whether an individual traveled at least once, where they traveled to and from, and any idiosyncratic travel patterns and routines they have may have, among others. As the number of days used in the simulation increases, this releases additional O-D matrices that adversaries can use to infer personal information about individuals. The mathematics of $\epsilon$-differential privacy formalize this intuition, and indicates that the likelihood an individual's data is compromised increases linearly in the amount of information they contribute. For example, in the pandemic simulation in Afghanistan, 305 private mobility matrices are released. When trip-level privacy is used, the privacy loss for an individual over the entire simulation is $\epsilon(n_1+...+n_{305})$, where $n_j$ is the number of inter-region trips an individual took on day $j$ that appear in the CDR. Likewise, when trip-level privacy is used in our case studies of the Lake Kivu Earthquake in Rwanda and the Battle of Kunduz in Afghanistan, where 7 days of O-D matrices are released, the privacy loss over the entire simulation is $\epsilon(n_1+...+n_{7})$, where $n_j$ is the number of inter-region trips an individual took on day $j$ that appear in the CDR. In both cases, the worst-case privacy loss is defined as the maximum privacy loss an individual incurs in a simulation. However, our empirical results indicate that expected privacy loss is substantially lower.

In both case studies, the expected privacy loss across the population is much lower than the worst-case characterization. In our pandemic case study in Afghanistan, the average individual makes a total of $n_1+...+n_{305} = 14$ province-to-province trips and $n_1+...+n_{305} = 52$ district-to-district trips during the 305 days of the simulation; these trips  are aggregated in mobility matrices (Table \ref{table:data}), so the average individual-level privacy loss is $14\epsilon$ at the admin-2 level and $52\epsilon$ at the admin-3 level for the entirety of the simulation. For a value of $\epsilon = 0.5$, this corresponds to a total privacy loss of $7$ for provinces and $26$ for districts; the privacy loss accordingly shrinks with smaller values of $\epsilon$. This means that, using the 305 released private O-D matrices, the information an adversary gains about an individual (such as whether they traveled, and if so, to and from where) is, on average, at most 7 times higher at the province level (and at most 26 times higher at the district level) than had their data not been used at all. (For comparison, in the non-private approach, the information an adversary gains about the average individual is, in principle, unbounded; but using differential privacy, privacy loss can be limited by the choice of $\epsilon$ and the number of O-D matrices released.)
In our case study of the Lake Kivu earthquake in Rwanda, only seven days of mobility matrices are released, and the average individual makes 1.53 trips at an admin-2 level and 2.88 trips at an admin-3 level (Table \ref{table:data}). For a value of $\epsilon = 0.5$, this corresponds to a total privacy loss of $0.77$ for provinces and $1.44$ for districts. Privacy loss in our case study of the Battle of Kunduz in Afghanistan is expected to be even smaller: over the seven days of data released, the average individual makes 0.26 trips at an admin-2 level and 1.73 trips at an admin-3 level (Table \ref{table:data}). For a value of $\epsilon = 0.5$, this corresponds to a total privacy loss of $0.13$ for provinces and $0.87$ for districts.

While there are few benchmarks against which we can compare the expected privacy loss of our algorithm, the numbers we observe are within the range of privacy losses observed in existing  differential privacy deployments. For instance, Google's COVID-19 Community Mobility Reports incurs a total privacy loss of 2.64 per day \cite{aktay2020google}. While Google users are likely more mobile than the individuals in our dataset, this provides a rough point of comparison. If Google's COVID-19 Community Mobility Reports were released for seven days, as in our disaster response scenarios, the total privacy loss would accumulate to 18.48. If Google's COVID-19 Community Mobility Reports were released for 305 days, as in our pandemic response scenario, the total privacy loss would accumulate to 805.2.   

To further stress-test our method, we implement a state-of-the-art membership inference attack from Pyrgelis et al. (2017) \cite{pyrgelis2017knock}, which uses machine learning tools to attempt to predict whether an individual's information is included in an aggregate dataset (see Supplementary Information). We find that the raw (non-private) mobility matrices we test do indeed leak information about membership of individuals in the construction of the matrix: machine learning classification models achieve an average AUC of 0.62-0.63 in a held-out test set (Table S1). The risk of membership inference implied by these AUC values is low in absolute terms and in comparison to prior work\cite{pyrgelis2017knock}, but nonetheless indicates that attackers could gain access to individual information through the non-private mobility matrix. 
When our method is used to privatize the mobility matrix prior to the membership inference attack, the attacker's accuracy is lowered in some, but not all, cases (AUC decrease of 0-3\% in Table S1). These modest decreases in AUC can have meaningful implications at scale: for instance, roughly 100,000 fewer individuals could be identified as being more likely to exist in the O-D matrix than a random individual (of the 7.12 million individuals in the entire dataset). The protections provided by our method are less than those of the empirical tests of differential privacy in Pyrgelis et al. (2017), in part because the baseline leakage of individual information (absent privatization) is much lower. Our setting is also challenging because we simulate providing the adversary with considerably more data per individual than in prior work\cite{pyrgelis2017knock}.  

\subsubsection*{Methods to tune the privacy-accuracy tradeoff} 
How should policymakers choose a value of $\epsilon$ to navigate the privacy-accuracy tradeoff? Across applications of differential privacy, there is no  consensus on the ``right'' way to set $\epsilon$ \cite{dwork2019differential, naldi2015differential, hsu2014differential, kohli2018epsilon}. In humanitarian applications, where accurate statistics can directly improve individual and social welfare, one realistic approach is to maximize the amount of privacy an algorithm can provide subject to accuracy constraints that are determined by programmatic requirements \cite{ligett2017accuracy, whitehouse2022brownian}. 

In situations where policymakers have a maximum tolerance for accuracy loss, one might imagine using estimates of intervention accuracy loss from the datasets being released to set $\epsilon$. For instance, consider an aid organization seeking to target relief to the top-$3$ Admin-3 locations with the highest out-migration during the Battle of Kunduz. Suppose the aid organization wants to provide the strongest privacy protections possible, provided the algorithm has at least $90\%$ intervention accuracy. The aid organization generates Table \ref{table:targeting} and chooses the smallest value of $\epsilon$ that yields $90\%$ accuracy, namely $\epsilon = 0.5$. Then, the aid organization uses our algorithm with $\epsilon = 0.5$ to generate private O-D matrices to determine the 3 locations to receive relief. (Note that this approach, while intuitively appealing, has a subtle caveat: technically, the parameters of an algorithm, such as the value of $\epsilon$, can leak privacy when chosen based on the data\cite{kohli2018epsilon, liu2019private, papernot2021hyperparameter, cummings2024advancing}. One approach to address this would be to compute private mobility matrices on pre-program data and set $\epsilon$ based on simulations on those data, before applying the chosen $\epsilon$ to the deployment data).

Alternatively, policymakers could set $\epsilon$ using heuristics based on our algorithm \cite{kohli2023differential}. For example, one simple heuristic approach is to determine $\epsilon$ based on a policymaker's tolerance for error of an O-D matrix count based on the noise introduced by Step 2 of Algorithm \ref{algo:od}. The standard deviation of the Laplace distribution with $\lambda = \epsilon^{-1} T$ is given by $T\sqrt{2}/\epsilon$, so we can bound on the set of feasible $\epsilon$ for a given O-D matrix count accuracy $\alpha$, as $\alpha \ge \sqrt{2}T/\epsilon \iff \epsilon \ge \sqrt{2} T/ \alpha$. For example, if policymakers require that the typical error in each entry of the O-D matrix is at most 10 (trip-level privacy), then they can set $\epsilon \ge \sqrt{2} / 10 \approx 0.14$. However, if policymakers instead can tolerate a typical error of at most 50, then any $\epsilon \ge \sqrt{2} / 50 \approx 0.028$ will suffice. While simple, this heuristic has two drawbacks: First, it ignores the impact of the post-processing steps of our algorithm, which can affect the overall accuracy of our algorithm \cite{zhu2021bias, kohli2021leveraging, wang2021post} and downstream applications \cite{pujol2020fair, cohen2022private, zhu2022post}; And second, this heuristic only utilizes the first two moments of the Laplace distribution to find a rough balance between the privacy level $\epsilon$ and the typical error. 

A final option is to use the privacy-accuracy tradeoff theorems provided in the \textit{Supplementary Information} to configure the algorithm to behave optimally. Theorems 6 and 7 (\textit{Supplementary Information}) provide explicit formulas for policymakers to set $\epsilon$ based on context-specific needs. For example, suppose policymakers require $95\%$ confidence that the error in an entry of the private O-D matrix be at most 10. Using Theorem 6, the policymakers can utilize any $\epsilon \ge -(10.5)^{-1}\ln(0.05) \approx 0.285$. Compared to the value of 0.14 from the simple heuristic above, Theorem 6 enables policymakers to achieve a similar goal of bounding errors in each cell by 10 while utilizing a larger $\epsilon$, thereby enabling the construction of more accurate private matrices. Theorem 7 can be similarly utilized to readily translate a policymaker's tolerance for error in terms of the parameter $\epsilon$.

In addition to $\epsilon$, the choice of the suppression threshold $\tau$ can also affect the privacy-accuracy tradeoff. As is the case with $\epsilon$, there is no single way to set the suppression parameter. Historically, different organizations have used varying suppression thresholds based on their perception of privacy risk. For example, the US Centers for Disease and Control and Prevention use a suppression threshold of 15 for their cancer statistics\cite{cdc2023}, whereas California's Department of Education uses a threshold of 10 for their public education statistics\cite{california}. In situations where an organization already has a suppression standard $s$, it can be used to set $\tau$ based on the noise introduced by our algorithm. (Since the threshold only applies to Step 3 of the algorithm, it is sufficient to consider the impact of the noise from Step 2 of the algorithm). We consider two natural heuristics. The first heuristic leaves the suppression unchanged ($\tau = s$), as the expected noise introduced by the Laplace distribution is 0. While this heuristic allows $\tau$ to account for the average noise introduced, it does not capture the variability of the noise introduced by the Laplace distribution. Thus, a second heuristic sets $\tau = s \pm \sqrt{2}/\epsilon$, where the $+$ or $-$ is chosen based on domain-specific concerns. For example, if the policymaker's intervention accuracy requires few cells to be suppressed, they can set $\tau = s - \sqrt{2}/\epsilon$ and provide additional privacy protections to individuals via data minimization. Alternatively, if the policymaker is concerned that small sub-populations could now appear in the private O-D matrix due to the introduction of noise,  they could set $\tau = s + \sqrt{2}/\epsilon$ to reduce the likelihood of such an event.

\section*{Discussion}
\label{sec:discussion}
In this paper, we presented a differentially private algorithm for releasing O-D matrices to learn about the mobility of a population without learning ``too much'' about any individual (or any individual trip). We analytically derived closed-form expressions for the privacy-accuracy tradeoff and showed that our algorithm can be configured to produce accurate private mobility matrices for a wide-range of parameter settings. 

We then tested the performance of our differentially private O-D matrix algorithm using mobility data from two humanitarian settings. In these settings, we compare policy decisions made using the original mobility data to decisions made from a privatized version of our data. This allows us to empirically characterize the privacy-accuracy tradeoff, and calibrate privacy in a way that allows the policymaker to still make effective decisions. Our closing discussion provides more general and practical guidance for calibrating privacy (i.e., setting $\epsilon$) based on the policymaker's tolerance for uncertainty. 

As digital devices proliferate throughout the developing world, the data generated by those devices are increasingly being used to advance the Sustainable Development Goals. In these and other ``social good'' applications, it is imperative to consider and respect the privacy of the individuals behind the data. Our algorithm, analysis, and discussion highlights one way to provide strong privacy guarantees while still allowing for policymakers to make informed decisions. We hope future work can improve and expand upon these methods, to provide a more robust set of options for using private data in effective humanitarian response.

\section*{Methods}
\label{sec:methods}
\subsection*{Differentially private O-D matrices}

The main text provides an intuitive description of the differentially private O-D matrix algorithm. Here we formalize the algorithm, and describe a number of accuracy and privacy guarantees. We defer the formal presentation of these statements, as well as their proofs, to the \textit{Supplementary Information}.

Our Private O-D Matrix algorithm (Algorithm \ref{algo:od}) requires four inputs: the dataset $d$ which is used to construct the private O-D matrix, a value of $\epsilon$ to control the privacy-accuracy tradeoff in the private matrix, a value $T$ which has two interpretations based on whether trip-level differential privacy or individual-level differential privacy is required, and a threshold parameter $\tau \in \Z_{\ge 0}$ that ensures the noisy matrix entries are non-negative and enables suppression of small counts when desired. We describe these last two parameters in more detail below.

When individual-level differential privacy is required, $T$ represents the maximum number of trips that any one individual can contribute to the mobility matrix. In order to ensure that the dataset $d$ does not contain more than $T$ trips for any individual, a preprocessing transformation can be applied to censor additional trips. One such algorithm that does so runs as follows. For each individual, the algorithm determines the number of trips $\theta$ they took. If $\theta > T$, the algorithm selects $T$ trips uniformly at random to keep and drops the remaining $\theta - T$ trips from the dataset. If $\theta \le T$, the algorithm leaves all $\theta$ trips in $d$. Alternatively, when trip-level differential privacy is required, setting $T=1$ achieves this requirement, as the presence or absence of any trip changes any count in the O-D matrix by 1. 

For any choice of $\tau \in \Z_{\ge 0}$, the entries of the private matrix are guaranteed to be non-negative. The magnitude of $\tau$ determines the level of suppression present in the private output. When $\tau = 0$, there is no suppression of the values in the matrix. As we increase $\tau$, we increase the set of small counts to be suppressed. Prior to provable privacy techniques, it was common practice in statistical disclosure control to suppress information derived from a small number of individuals \cite{matthews2017review}. While suppression alone offers no provable privacy guarantee, it can inform end-users that the statistic may be unreliable due to the limitations in the data collection methods, such as having low phone coverage in a particular region. Additionally, for certain types of humanitarian responses -- such as those involving disease spread -- regions with small mobility counts are likely to provide fewer vectors for disease spread (and hence could be less relevant for mitigation strategies) compared to those regions with higher mobility counts. Consequently, suppression may be a useful signalling apparatus to policymakers.

\begin{algorithm}
\KwIn{
\begin{itemize}
    \item Dataset $d$
    \item Privacy parameter $\epsilon \ge 0$
    \item Trip threshold  $T \in \Z_{\ge 1}$: For individual-level differential privacy, $T$ corresponds to the maximum number of trips an individual can contribute to the non-private O-D matrix; the dataset $d$ should be preprocessed to ensure every individual contributes at most $T$ trips. For trip-level differential privacy, set $T = 1$; no additional preprocessing of $d$ is required.
    \item Suppression threshold $\tau \in \Z_{\ge 0}$
\end{itemize}
}
\KwOut{$\epsilon$-differentially private O-D matrix $\hat{M}(d)$}
(\textit{Non-private computation}) For a geographic region with $k$ zones, compute the $k \times k$ matrix $M(d)$, where  $M(d)_{a,b}$ equals the number of trips from zone $a$ to zone $b$ when $a \ne b$, and 0 when $a = b$.\\
(\textit{Incorporate privacy}) Add independent Laplace noise with scale $\lambda = \epsilon^{-1}T$ to each non-diagonal entry of $M(d)$. Call this resulting matrix $M'(d)$.\\
(\textit{Post-processing}) Next, \linebreak
(I) Round each of the resulting entries in $M'(d)$ to the nearest integer. In the event where there is no unique nearest integer, we adopt the convention of rounding up to the nearest integer.
\linebreak
(II) Then, for all resulting values that are below a threshold of $\tau$, map these values to 0.
\linebreak
Call this resulting matrix $\hat{M}(d)$ \\
\textbf{Return} $\hat{M}(d)$
\caption{Private O-D Matrix}
\label{algo:od}
\end{algorithm}

This algorithm is accompanied with a several provable guarantees. For starters, our algorithm satisfies $\epsilon$-differential privacy (Theorem 1, \textit{Supplementary Information}). In addition to the privacy guarantee, our algorithm has strong accuracy guarantees, which we call the \textit{privacy-accuracy tradeoff theorems}. With high probability: (1) the algorithm preserves cells that would have been suppressed absent differential privacy; (2) the algorithm does not suppress cells that would not have been suppressed absent differential privacy; (3) for non-suppressed cells, the algorithm's error decays exponentially as $\epsilon$ increases; and (4) for non-suppressed cells, when examining the difference in cell counts across different time periods, the algorithm's error decreases according to $\Theta(\epsilon\exp(-\epsilon))$ as $\epsilon$ increases. These statements are made rigorous in Theorems 2, 3, 4, and 5 respectively in the \textit{Supplementary Information}. Theorems 6 and 7 utilize Theorems 4 and 5 to provide policymakers with exact formulas to calculate $\epsilon$ based on their context-specific needs.

\subsection*{Empirical Simulations}

\subsubsection*{Data}

Our simulations that use the private O-D matrix algorithm for real-world policy decisions rely on location traces derived from mobile phone metadata (call detail records, or CDR). We use three call detail record datasets: one from Afghanistan in 2015, one from Afghanistan in 2020, and one from Rwanda in 2008. Table \ref{table:data} summarizes these data. Our dataset from Afghanistan in 2020 is 305 days long and contains transactions from around 7 million subscribers; the remaining two datasets are only 7 days long (as the focus on specific natural disasters or violent events) and contain data for half a million (Rwanda) and around three million (Afghanistan) subscribers, respectively. In each dataset the CDR include the date, time, and duration of each call placed on the mobile phone networks, along with pseudonymized IDs for the originating subscriber and the recipient. They also record the cell tower through which each call was placed, providing a measure of the approximate location of each originating subscriber. We derive daily O-D matrices at the admin-2 and admin-3 level for each dataset by geolocating cell towers to administrative subdivisions, and then counting the number of trips observed in the data on a daily basis (where a trip occurs if a subscriber places a call or text in one administrative subdivision and a subsequent call or text in a different subdivision). 

\subsubsection*{Algorithm parameters}

Across all simulations, we use a suppression threshold of 15, which is used in practice by the US Centers for Disease and Control and Prevention in their cancer statistics \cite{cdc2023}, as well as the US Department of Health and Human Services National Institute on Minority Health and Health Disparities \cite{nih}. For information on the share of O-D matrix counts that are suppressed in each datasets, see Table S2 in \textit{Supplemental Information}. We set $T = 1$ to provide trip-level protections with $\epsilon \in \{0.1, 0.5, 1\}$. We also vary the value of $T$ in the geographic targeting simulations to quantify the impact that individual-level privacy has on the utility of downstream policy decisions (described below).

\subsubsection*{Non-pharmaceutical interventions in pandemic response}
For this case study, we rely on the mobility-based SIR (Susceptible, Infected, and Recovered) model introduced by Goel and Sharma \cite{goel2020mobility}, which adapts the classic SIR model of Kermck and McKendrick \cite{kermack1927contribution} by including features to model human mobility. In the mobility-based SIR model, at any time $t$, individuals in region $i \in [k]$ are categorized into one of three groups: they are either susceptible to the disease, infected with the disease, or recovered from the disease. The number of individuals in these groups is given by $S_i(t), I_i(t),$ and $R_i(t)$ respectively. We denote the population of region $i$ at time $t$ as $N_i(t) = S_i(t) + I_i(t) + R_i(t)$. As time goes on, the number of individuals in these compartments can change. The epidemic dynamics are a function of both intra- and inter-region disease spread, described by the following differential equations involving a $k \times k$ O-D matrix $M$, that are governed by three parameters: $\beta$, $\alpha$, and $\mu$.
\[
\begin{split}
    & \frac{dS_i(t)}{dt} = -\frac{\beta S_i(t)I_i(t)}{N_i(t)} - \frac{\alpha \beta S_i(t) \sum_{j \in [k]} \frac{M_{i, j} I_j(t)}{N_j(t)}}{N_i(t) + \sum_{j \in [k]} M_{i, j}} \\
    & \frac{dR_i(t)}{dt} = \mu \frac{I_i(t)}{N_i(t)} \\
    & \frac{dI_i(t)}{dt} = \frac{dS_i(t)}{dt} - \frac{dR_i(t)}{dt} \\
\end{split}
\]
The parameter $\beta$ describes the frequency of mixing of intra-region populations in combination with the virality of the disease: how likely is an infected person to meet a non-infected person from the same region and transfer the pathogen? The $\alpha$ term in the differential equations describes the frequency of mixing with inter-region visitors in relation to the frequency of mixing with intra-region residents: how likely is an infected person to meet a visitor and transfer the pathogen? The last parameter $\mu$ describes the rate of recovery from disease: what proportion of the infected population recovers each day? 

In principle, $\beta$, $\alpha$, and $\mu$ could all vary from region to region on the basis of a number of factors (for example, community cohesion, access to healthcare in the area, and mobility restrictions), but for simplicity we keep them constant across regions. To configure this model, we set $\alpha=1$, meaning that there is no difference between the degree of mixing with residents and the degree of mixing with visitors. We additionally set $\beta$ and $\mu$ based on estimated dynamics from the COVID-19 pandemic. SIR estimates of $\beta$ range between 0.06 and 0.39 in SIR models fit in 2020 in a set of nations; estimates of $\mu$ range between 0.04 and 0.19 in the same set of studies \cite{bertozzi2020challenges, bagal2020estimating, lounis2020estimation}. For this case study, we set $\beta=0.10$ and $\mu=0.04$.

The mobility-based SIR model relies additionally on mobility matrices $M$ and a time-invariant regional population values $N_1,...,N_k$. To compare the intervention differences between the private and non-private regimes, we constructed non-private mobility matrices using the CDR from Afghanistan along with private versions with $\epsilon \in \{0.1,0.5,1\}$ using our private O-D matrix algorithm at the trip level. We estimate the regional population values $N_1,...,N_k$ by inferring a home location for each subscriber in the CDR dataset using the most common cell tower through which they placed calls between the hours of 8pm and 6am. These regional counts $N_1,...,N_k$ could leak information about individuals, so in the differentially private version of the pandemic simulation we privatize the regional population counts $N_1,...,N_k$ with the same value of $\epsilon$ as the OD matrices. 

As described in the results section, we use the mobility-based SIR model with daily origin-destination matrices to calculate epidemic curves (Figure S1, \textit{Supplementary Information}), and compare the accuracy of policy decisions taken when SIR models are deployed with private and non-private OD matrices.  

\subsubsection*{Geographic targeting of humanitarian aid after natural disasters and violent events}
Our second simulated policy context is the targeting of humanitarian aid after natural disasters, violent events, and other shocks. As described in the results section, we calculate daily O-D matrices for time periods associated with two such shocks in our datasets: the Battle of Kunduz in Afghanistan in 2015, and the Lake Kivu Earthquake in Rwanda in 2008. We calculate daily O-D matrices for seven days following each event and calculate the two statistics using both private and non-private O-D matrices: (1) the total amount of out migration from the affected area, and (2) the top-$k$ districts with the most out migration. We assess the accuracy of both statistics using private data in comparison to non-private data. We also assess the sensitivity of these results to the choice of $k$ (Figure S3, \textit{Supplementary Information}). To show the impact the utility of downstream policy decisions using individual-level privacy, we also perform geographic targeting simulations by selecting different values of $T$ (Table S3, \textit{Supplementary Information}).

\section*{Data availability}
\label{sec:data_availability}
The mobile phone datasets from Afghanistan and Rwanda contain detailed information on billions of 
mobile phone transactions. These data contain proprietary and confidential information belonging to a private telecommunications operator and cannot be publicly released. Upon reasonable request, we can provide information to accredited academic researchers about how to request the proprietary data from the telecommunications operator.  Contact the corresponding author for any such requests.
\bibliography{sample}

\section*{Acknowledgements}

 This work was supported by CITRIS and the Banatao Institute at the University of California, the Bill and Melinda Gates Foundation via the Center for Effective Global Action, DARPA and NIWC under contract N66001-15-C-4066, and the National Science Foundation under CAREER Grant IIS-1942702. Aiken further acknowledges funding from a Microsoft Research PhD Fellowship. The U.S. Government is authorized to reproduce and distribute reprints for Governmental purposes not withstanding any copyright notation thereon. The views, opinions, and/or findings expressed are those of the author(s) and should not be interpreted as representing the official views or policies of the Department of Defense or the U.S. Government. We are grateful to Sveta Milusheva and Robert Marty for providing early feedback on this work, as well as Paul Laskowski and Deirdre Mulligan for helpful discussions.

\section*{Author contributions statement}
NK and JB conceived of the research idea. NK developed the private mobility matrix algorithm and proved each of the theorems in the paper. EA implemented the empirical measurements of accuracy and simulations of intervention accuracy. NK, EA, and JB wrote the manuscript.

\section*{Additional information}

The author(s) declare no competing interests.

\clearpage

\section*{Figures and Tables}
\begin{figure}[!h]
\centering
\includegraphics[width=\textwidth]{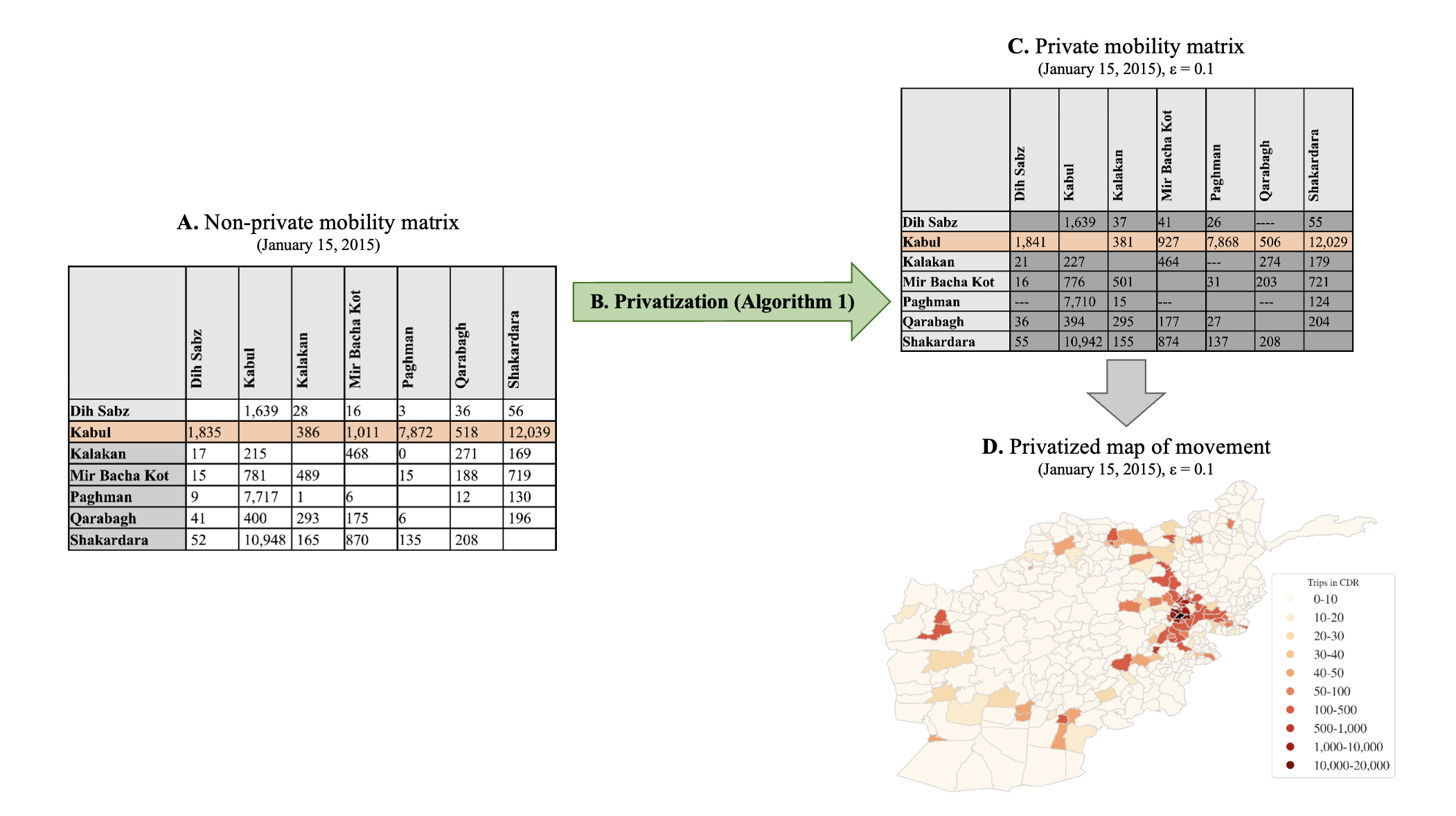}
\caption{Panel A: A portion of an origin-destination matrix, showing movement between 10 districts of Afghanistan on January 1, 2015, as calculated from data provided by a mobile phone operator. Panel B: The origin-destination matrix is passed into our private O-D matrix algorithm (as described in Algorithm \ref{algo:od}) for trip-level protection ($T$ =1) with privacy parameter $\epsilon = 0.1$ and suppression threshold $\tau = 15$. Panel C: The same portion of the origin-destination matrix, this time after being privatized by our private O-D matrix algorithm (with $\epsilon = 0.1$). Panel D: Movement from Kabul to other districts of Afghanistan on January 1, 2015, based on the highlighted row of the private origin-destination matrix in Panel B. Kabul, the origin district, is shown in black.}
\label{fig:matrix}
\end{figure}

\begin{figure}[!h]
\centering
\includegraphics[width=\textwidth]{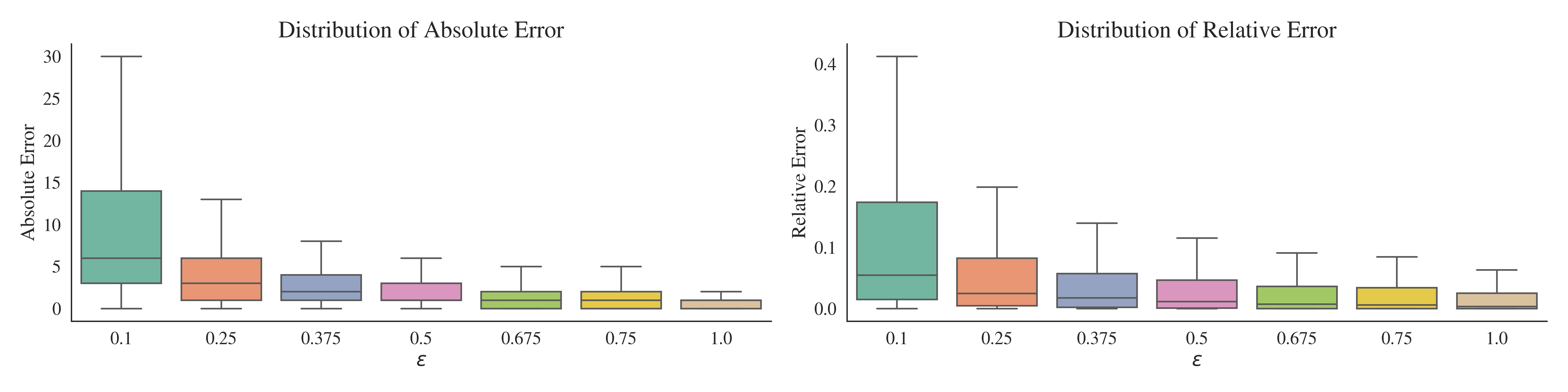}
\caption{Left: Absolute errors in matrix entries for differentially private O-D matrices, relative to non-private matrices (derived over all days in our Afghanistan 2020 dataset and all origin-destination pairs at the admin-2 or province level). Absolute error is calculated as the difference between an O-D matrix count in the private matrix and the corresponding O-D matrix count in the non-private matrix. The distributions shown are taken over all 305 days in our 2020 CDR dataset (see Table \ref{table:data}) and all origin-destination province pairs. Right: Relative errors in matrix entries for differentially private O-D matrices, relative to non-private matrices. Relative error is calculated as twice the absolute error in a private O-D matrix count, divided by sum of the private and non-private O-D matrix counts for the same cell. Again, the distributions shown are taken over all 305 days in our 2020 CDR dataset (see Table \ref{table:data}) and all origin-destination province pairs.}
\label{fig:boxplots}
\end{figure}

\begin{table}[!h]

\caption{Summary of the call detail record datasets used in our empirical simulations.}
\label{table:data}
\centering
\small
\begin{tabular}{@{}llll@{}}
\toprule
\textbf{}                                                        & \textbf{Rwanda 2008}          & \textbf{Afghanistan 2015}      & \textbf{Afghanistan 2020}    \\ \midrule
\textbf{Dates}                                                   & Feb. 3 - Feb. 9, 2008 & Sep. 28 - Oct. 4, 2015 & Jan. 1 - Oct. 31, 2020 \\
\textbf{Number of days}                                          & 7                             & 7                              & 305                          \\
\textbf{Number of transactions}                                  & 13 million                    & 64 million                     & 3.2 billion                  \\
\textbf{Number of subscribers}                                   & 541 thousand                  & 2.79 million                   & 7.12 million                 \\
\textbf{Admin-2 level}                                           & 30 districts                  & 34 provinces                   & 34 provinces                 \\
\textbf{Admin-3 level}                                           & 416 sectors                   & 421 districts                  & 421 districts                \\
\textbf{5th percentile of trips per subscriber (admin-2 level)} & 0.00                            & 0.00                                & 0.00                             \\
\textbf{5th percentile of trips per subscriber (admin-3 level)} & 0.00                              & 0.00                               & 0.00                            \\
\textbf{Mean trips per subscriber (admin-2 level)}               & 1.53                          & 0.26                           & 14                           \\
\textbf{Mean trips per subscriber (admin-3 level)}               & 2.88                          & 1.73                           & 52.00                        \\
\textbf{95th percentile of trips per subscriber (admin-2 level)} & 10                            & 2                              & 52                           \\
\textbf{95th percentile of trips per subscriber (admin-3 level)} & 20                            & 9                              & 251                          \\ \bottomrule
\end{tabular}

\end{table}

\begin{table}[!h]

\caption{Average errors in pandemic response (measured in accuracy, precision, recall) by degree of privacy introduced in the mobility matrix, at the admin-2 (left) and admin-3 level (right). Accuracy, precision, and recall are calculated for each province/district across all 305 timesteps in the simulation, and averages are taken across all provinces/districts. Standard deviations across provinces/districts are reported in parentheses. }
\label{table:pandemic}
\centering
\small
\begin{tabular}{@{}lllllll@{}}
\toprule
               & \multicolumn{3}{c}{\textbf{Admin-2 Level (Provinces)}}     & \multicolumn{3}{c}{\textbf{Admin-3 Level (Districts)}}      \\ 
               & Accuracy   & Precision   & Recall      & Accuracy    & Precision   & Recall      \\ \midrule
\multicolumn{7}{l}{\textit{Panel A: Pandemic initiating in Kabul}}                                         \\
Non-Private     & 100\%      & 100\%       & 100\%       & 100\%       & 100\%       & 100\%       \\
$\epsilon$=0.1 & 98\% (2\%) & 94\% (10\%) & 95\% (10\%) & 83\% (11\%) & 61\% (27\%) & 64\% (27\%) \\
$\epsilon$=0.5 & 99\% (0\%) & 98\% (3\%)  & 98\% (3\%)  & 91\% (8\%)  & 79\% (21\%) & 80\% (20\%) \\
$\epsilon$=1.0 & 99\% (0\%) & 98\% (2\%)  & 98\% (2\%)  & 91\% (9\%)  & 78\% (22\%) & 80\% (22\%) \\ \midrule
\multicolumn{7}{l}{\textit{Panel B: Pandemic initiating in Hirat}}                                         \\
Non-Private     & 100\%      & 100\%       & 100\%       & 100\%       & 100\%       & 100\%       \\
$\epsilon$=0.1 & 90\% (8\%) & 64\% (33\%) & 63\% (33\%) & 64\% (12\%) & 16\% (27\%) & 18\% (29\%) \\
$\epsilon$=0.5 & 96\% (5\%) & 85\% (21\%) & 85\% (21\%) & 72\% (12\%) & 33\% (29\%) & 35\% (30\%) \\
$\epsilon$=1.0 & 97\% (5\%) & 89\% (20\%) & 89\% (20\%) & 77\% (11\%) & 44\% (29\%) & 45\% (29\%) \\  \midrule
\multicolumn{7}{l}{\textit{Panel C: Pandemic initiating randomly}}                                         \\
Non-Private     & 100\%      & 100\%       & 100\%       & 100\%       & 100\%       & 100\%       \\
$\epsilon$=0.1 & 93\% (8\%) & 74\% (31\%) & 74\% (32\%) & 71\% (15\%) & 33\% (36\%) & 34\% (37\%) \\
$\epsilon$=0.5 & 97\% (5\%) & 89\% (19\%) & 89\% (19\%) & 79\% (13\%) & 52\% (31\%) & 53\% (32\%) \\
$\epsilon$=1.0 & 98\% (3\%) & 93\% (12\%) & 93\% (12\%) & 82\% (12\%) & 58\% (30\%) & 59\% (30\%) \\ \bottomrule
\end{tabular}

\end{table}

\begin{table}[!h]
\caption{Average errors in total migration counts and top-$k$ regions of out-migration, by degree of privacy introduced in the mobility matrix, at the admin-2 level (left) and admin-3 level (right).}
\label{table:targeting}
\centering
\small
\begin{tabular}{@{}lllllllll@{}}
\toprule
                                     & \multicolumn{4}{c}{\textbf{Admin-2 level}} & \multicolumn{4}{c}{\textbf{Admin-3 level}} \\ \midrule
                                     & \textit{Non-Private}  & $\epsilon=0.1$ & $\epsilon=0.5$ & $\epsilon=1$ & \textit{Non-Private}  & $\epsilon=0.1$ & $\epsilon=0.5$ & $\epsilon=1$ \\ \midrule 
\multicolumn{7}{l}{\textit{Panel A: Battle of Kunduz in Afghanistan}}                                                          \\
Total out-migration                  & 49,994           & 49,001  & 48,725     & 48,722    & 87,007           & 83,327    & 82,671      & 82,609    \\
Percent error in total out-migration & 0.00\%               & 1.98\%   & 2.54\%      & 2.61\%    & 0.00\%               & 4.23\%      & 4.98\%   & 5.05\%    \\
Accuracy of top-$k$ regions ($k=3$)      & 100.00\%          &          90.47\%     & 90.47\%   & 90.47\%   & 100.00\%               & 85.71\%    & 90.47\%    & 90.47\%   \\ \\
\multicolumn{7}{l}{\textit{Panel B: Lake Kivu Earthquake in Rwanda}}                                                           \\
Total out-migration & 32,627           & 30,595    & 29,930   & 29,801                 & 51,102           & 48,085    & 47,331     & 47,320        \\
Percent error in total out-migration &    0.00\%               & 6.23\%   & 8.27\%      & 8.66\% & 0.00\%                    & 5.90\%    & 7.38\%     & 8.66\%      \\
Accuracy of top-$k$ regions ($k=3$)  & 100.00\%               & 90.47\%   & 95.24\%     & 90.47\%    & 100.00\%                    & 100.00\%     & 100.00\%     & 100.00\%        \\ \bottomrule
\end{tabular}
\end{table}


\clearpage 
\section*{Supplementary Information}
\section*{Mathematical Details: Differential Privacy}
In this section, we provide the mathematical foundation of Private O-D Matrix Algorithm, as well as formalize and prove all the theorems referenced in the main manuscript.

\subsection*{Mathematical Preliminaries} 
\label{math_prelims}

Differential privacy provides a mathematical guarantee that a privacy-preserving statistic $A(d)$ that approximates a non-private statistic $M(d)$ derived from dataset $d$ will not reveal ``too much" information about any underlying observation in dataset $d$ \cite{dwork2019differential}.

\begin{defn}\label{def:dp} ($\epsilon$-Differential Privacy \cite{dwork2006calibrating})
A randomized algorithm $A$ mapping datasets in $\D$ to an outcome space $\O$ satisfies $\epsilon$-differential privacy if, for all datasets $d,d' \in \D$ that differ in one element, and for all events $O \subseteq \O$, $\P(A(d) \in O) \le \exp(\epsilon)\P(A(d') \in O)$.
\end{defn}

Thus the privacy preserving statistic $A(d)$ is  the output of a randomized algorithm $A$ whose noisy-behavior is governed by the parameter $\epsilon$. 

It is worth noting that differential privacy is not a particular technique. Rather, it is a mathematical standard that an algorithm either satisfies or not \cite{dwork2019differential}. As such, differential privacy provides us with a promise about the worst-case privacy loss incurred from running a computational process. As a general matter however, it does not tell us how to implement the promise. There are many possible ways to construct an algorithm to achieve the privacy guarantee. A foundational technique for privatizing a statistic $S: \D \rightarrow \R^m$ is to use noisy values sampled from a Laplace distribution\footnote{The Laplace distribution with scale parameter $\lambda \in (0,\infty)$ corresponds to a real-valued random variable with probability density function $f(z) = \frac{1}{2\lambda}\exp(-|z| / \lambda)$ \cite{dwork2006calibrating}. This distribution has mean 0 and a standard deviation of $\sqrt{2}\lambda$.} with a carefully crafted scale parameter $\lambda$ (denoted as $Lap(\lambda)$) to ensure that contribution of \textit{any} possible data point is masked.

\begin{lem} [Laplace Mechanism \cite{dwork2006calibrating}]
For a statistic $S:\D \rightarrow \R^m$, let $\Delta_S$ be the supremum of $||S(d) - S(d')||_1$, where $d,d'\in \D$ are datasets that differ in one element. For any $\Delta_S \in (0,\infty)$, the algorithm $A(d) = S(d) + (\zeta_1,...,\zeta_m)^T$, where each $\zeta_i$ is drawn independently from $Lap(\epsilon^{-1}\Delta_S)$, satisfies $\epsilon$-differential privacy.
\end{lem}

It is also worth highlighting that the definition of differential privacy quantifies over datasets that differ in one \textit{element}. There are two ways to consider this elemental difference in our setting, and each has a different consequence for the level of privacy afforded. When the elemental difference in mobility datasets is a trip, differential privacy affords \textit{trip-level privacy protection} by protecting each trip with parameter $\epsilon$. Hence, if an individual took $t$ trips, they experience $t\epsilon$ privacy loss. This can result in \textit{heterogeneous privacy loss} for data subjects, with individuals who travel more incurring larger privacy losses. Alternatively, when the elemental difference in datasets is an individual, differential privacy protects all the trips an individual took with parameter $\epsilon$, yielding \textit{homogeneous privacy loss} for all data subjects. This \textit{individual-level privacy protection} is a stronger privacy guarantee, but requires more noise in order to hide multiple trips.

Differentially private algorithms are accompanied with a host of provable guarantees that enable rigorous reasoning about privacy losses that occur from the computation of statistics. Three important properties for this study are the \textit{composition result}, the \textit{post-processing result}, and the \textit{group privacy}.

\begin{lem} [Composition \cite{dwork2014algorithmic}]\label{lem-composition}
Given any collection of $\epsilon_i$-differentially private algorithms $A_i:\D \rightarrow \O_i$, the algorithm $A = (A_1,...,A_r): \D \rightarrow \O_1 \times ... \times \O_r$ is ($\epsilon_1+...+\epsilon_r$)-differentially private. 
\end{lem}

\begin{lem} [Group-Privacy \cite{dwork2014algorithmic}]\label{lem-group}
Suppose $d,d'\in \D$ are datasets that differ in $g$ elements. Then any $\epsilon$-differentially private algorithm $A: \D \rightarrow \O$ further guarantees that $\P(A(d) \in O) \le \exp(g\epsilon)\P(A(d') \in O)$  for all events $O \subseteq \O$.
\end{lem}

\begin{lem} [Post-Processing \cite{dwork2014algorithmic}]\label{lem-post}
Given an $\epsilon$-differentially private algorithm $A: \D \rightarrow \O$ and any (potentially randomized) function $f: \O \rightarrow \O'$, $f \circ A$ is still $\epsilon$-differentially private.
\end{lem}

Practically speaking, if humanitarian response requires multiple computations to be run on the data, then the composition result allows us to keep track of the total privacy loss accumulated by adding up the $\epsilon$'s used. Additionally, if the output of a differentially private analysis subsequently needs to be analyzed or manipulated in any way without the original data, then the post-processing result ensures no additionally privacy loss has occurred. And lastly, the group privacy result informs us that: (1) when trip-level privacy is used, the total privacy incurred scales linearly with the number of trips $g$ that an individual took; and (2) when individual-level privacy is used, the total privacy incurred by a group of $g$ individuals scales linearly as well.

\subsection*{Private O-D Matrix Algorithm}
\label{od_alg_sec}

Using the mathematical tools above, we deduce that our algorithm satisfies $\epsilon$-differential privacy.

\begin{thm}
\label{4-thm-alg_satisfies_dp}
For any value $\tau \in \Z_{\ge 0}$ that is selected independent of data, the Private O-D Matrix algorithm satisfies (1) individual-level $\epsilon$-differential privacy for any $T \in \Z_{\ge 1}$ that is also determined independently of the underlying data, and (2) trip-level $\epsilon$-differential privacy when $T=1$.
\end{thm}

\begin{proof}
It is sufficient to privatize the non-diagonal entries of $M(d)$, as these are the only entries that are responsive to $d$. When individual-level privacy is required, the inclusion or deletion of any record in $d$ can change the sum of the non-diagonal entries of $M(d)$ by $\pm T$. And when trip-level privacy is required, any record in $d$ can change the sum of the non-diagonal entries of $M(d)$ by $\pm 1$ (which equals $\pm T$ as per the algorithm). Hence, $\Delta_M = T$. Adding $k^2-k$ independently sampled $Lap(\epsilon^{-1}T)$ values to the non-diagonal entries of $M(d)$ is equivalent to reshaping the non-diagonal entries of $M(d)$ as a vector $S(d) \in \R^{k^2 - k}$ and then utilizing the Laplace mechanism, yielding $\epsilon$-differential privacy. After this, there is no additional privacy loss by the post-processing result, as $\tau$ was set agnostic of $d$.
\end{proof}

\subsection*{Privacy-Accuracy Tradeoff Theorems}
\label{priv_acc_thms}

Our first two privacy-accuracy tradeoff theorems demonstrate that our algorithm preserves suppression and non-suppression of the non-private entries with high probability.

To do so, we consider a dataset $d$ that produces a non-private matrix $M(d)$ and a private matrix $\hat{M}(d)$ using some threshold $\tau$ and some privacy parameter $\epsilon$. Unless $\tau = 0$ we cannot directly compare these two outputs, as the private matrix will not contain any counts between $1$ and $\tau-1$ due to the suppression. For an apples-to-apples comparison, let $\bar{M}(d)$ be the result of mapping all entries in $M(d)$ that are below $\tau$ to 0. 

\begin{thm}
\label{4-thm-suppression_holds_whp}
For any $T \in \Z_{\ge 1}$, $\tau \in \Z_{\ge 0}$, and for any suppressed non-diagonal entry $(a,b)$ in $\bar{M}(d)$, $\hat{M}(d)_{a,b}$ is also suppressed with probability $1-0.5\exp(-\epsilon T^{-1}(\tau - 0.5 - M(d)_{a,b}))$.
\end{thm}

\begin{proof} 
$\bar{M}(d)_{a,b}$ is suppressed $\iff M(d)_{a,b} < \tau$. For $\eta \sim \text{Lap}(\epsilon^{-1}T)$, we have $\P(\hat{M}(d)_{a,b} = 0 ) = \P(\text{round}(M(d)_{a,b} + \eta) < \tau) = \P(\eta < \tau - 0.5 - M(d)_{a,b}) = 1-0.5\exp(-\epsilon T^{-1}(\tau - 0.5 - M(d)_{a,b}))$.
\end{proof}




\begin{thm} 
\label{4-thm-nonsuppression_holds_whp}
For any $T \in \Z_{\ge 1}$, $\tau \in \Z_{\ge 0}$, and for any non-suppressed non-diagonal entry $(a,b)$ of $\bar{M}(d)$,  $\hat{M}(d)_{a,b}$ is also not suppressed with probability $1-0.5\exp(\epsilon T^{-1}(\tau + 0.5 - M(d)_{a,b}))$.
\end{thm}

\begin{proof} 
$\bar{M}(d)_{a,b}$ is not suppressed $\iff M(d)_{a,b} \ge \tau$. For $\eta \sim \text{Lap}(\epsilon^{-1}T)$, $\P(\hat{M}(d)_{a,b} \ne 0 ) = \P(\text{round}(M(d)_{a,b} + \eta) \ge \tau ) = \P(\eta \ge \tau + 0.5 - M(d)_{a,b}) = 1-0.5\exp(\epsilon T^{-1}(\tau + 0.5 - M(d)_{a,b}))$.
\end{proof}




Our third privacy-accuracy tradeoff theorem quantifies the probability that the error between a non-private matrix entry and private matrix entry exceeds $\alpha$, in cases where the $(a,b)$ entry is non-suppressed in both the non-private and private matrices.

\begin{thm}
\label{4-thm-nonsuppression_static_accuracy}
Suppose $(a,b)$ is a non-suppressed non-diagonal entry of both $M(d)$ and $\hat{M}(d)$. Then for any $\alpha \in \Z_{\ge 0}$, 
the chance that $|\hat{M}(d)_{a,b} - M(d)_{a,b}| > \alpha$ is 
$
\exp(-\epsilon T^{-1}(\alpha + 0.5))
$.
\end{thm}

\begin{proof} 
For $\eta \sim \text{Lap}(\epsilon^{-1}T)$ and $\alpha \in \Z_{\ge 0}$, we have $\P(|\hat{M}(d)_{a,b} - M(d)_{a,b})| > \alpha) = \P(|\text{round}(\eta)| > \alpha) = 2\P(\eta \ge \alpha+0.5) = \exp(-\epsilon T^{-1}(\alpha + 0.5))$.
\end{proof}


The previous three results can be viewed as demonstrating ``static'' accuracy guarantees. However, mobility matrices can be used to inform policy decisions based on differences in population flow at two points in time. Our privatization method can still allow for strong accuracy guarantees, even when looking at such differences. For any non-diagonal entry $(a,b)$, let $M_{t_1}(d)_{a,b}$ and $M_{t_2}(d)_{a,b}$ be represent the non-private counts of trips from $a$ to $b$ at time periods $t_1$ and $t_2$ respectively. The following privacy-accuracy tradeoff theorem allows us to quantify the chance that the observed difference in the entry $(a,b)$ in private matrices and the actual difference the $(a,b)$ entry of non-private matrices is larger than $\alpha$.

\begin{thm}
\label{4-thm-nonsuppression_dynamic_accuracy}
Suppose the counts in $M_{t_1}(d)_{a,b}, M_{t_2}(d)_{a,b}, \hat{M}_{t_1}(d)_{a,b},$ and $\hat{M}_{t_2}(d)_{a,b}$ are all non-suppressed. 
Then for any $\alpha \in \Z_{\ge 0}$ and $T \in \Z_{\ge 1}$,  the chance that $|(\hat{M}_{t_2}(d)_{a,b} - \hat{M}_{t_1}(d)_{a,b}) - (M_{t_2}(d)_{a,b} - M_{t_1}(d)_{a,b})|>\alpha$ is given by 
$$
\exp(-\epsilon T^{-1} (\alpha+1)) \frac{\epsilon T^{-1} (\alpha+1) + 2}{2}
$$
\end{thm}

\noindent To streamline the proof of the Theorem \ref{4-thm-nonsuppression_dynamic_accuracy}, we first provide a helpful technical lemma.

\begin{lem}
\label{4-lem-difference_laplaces}
For any $\alpha \in \R$ and independent $\eta_1,\eta_2 \sim Lap(\epsilon^{-1}T)$, 

$$\P(\eta_2 - \eta_1 \ge \alpha) =  \exp(-\sgn(\alpha)\epsilon \alpha T^{-1}) \frac{\epsilon \alpha T^{-1} +2(\sgn(\alpha)+\mathbb{I}(\alpha=0))}{4}$$
\end{lem}

\begin{proof} 
Each $\eta_i \sim Lap(\epsilon^{-1}T)$, so they can be written as the difference of two independent Exponential random variables with rate parameter $\epsilon T^{-1}$. So $\eta_2 - \eta_1$ can be written as the difference of independent Gamma random variables $X'$ and $Y'$, each with shape parameter 2 and rate parameter $\epsilon T^{-1}$. To simplify calculations, we rewrite $X' = (\epsilon^{-1}T)X$ and $Y' = (\epsilon^{-1}T)Y$, where $X$ and $Y$ independent Gammas, each with shape 2 and rate 1. Hence, $\eta_2 - \eta_1 = (\epsilon^{-1}T)(X-Y)$. The pdf for a Gamma distribution with shape 2 and rate 1 is given by $f(t) = t \exp(-t)$. 
So, when $z < 0$,
$$
f_{X-Y}(z) = \int_{0}^{\infty} f_X(x)f_Y(x-z)dx = \frac{1}{4}\exp(z)(1-z)
$$
Since $X - Y$ is symmetric about $0$, when $z \ge 0$ we deduce from the above integration that $f_{X-Y}(z) = \frac{1}{4}\exp(-z)(1+z)$. Hence, $f_{X-Y}(z) = \frac{1}{4}\exp(-|z|)(1+|z|)$ for all $z \in \R$, so then
$$
\P(\theta_2 - \theta_1 \ge \alpha) = \P(X - Y \ge \epsilon \alpha T^{-1}) = \int_{\epsilon \alpha T^{-1}}^{\infty} \frac{1}{4}\exp(-|z|)(1+|z|)dz 
$$
When $\alpha \ge 0$, 
$$
\P(\theta_2 - \theta_1 \ge \alpha) = \int_{\epsilon \alpha T^{-1}}^{\infty} \frac{1}{4}\exp(-z)(1+z)dz = \exp(-\epsilon \alpha T^{-1}) \frac{\epsilon \alpha T^{-1} + 2}{4}
$$
And when $\alpha < 0$, 
$$
\P(\theta_2 - \theta_1 \ge \alpha) = \int_{\epsilon \alpha T^{-1}}^{0} \frac{1}{4}\exp(z)(1-z)dz + \frac{1}{2} = \exp(\epsilon \alpha T^{-1}) \frac{\epsilon \alpha T^{-1} - 2}{4}  
$$
Consolidating these cases delivers the advertised claim.
\end{proof}

\paragraph*{Proof of Theorem \ref{4-thm-nonsuppression_dynamic_accuracy}} 
\begin{proof}
Let $D_{1,2} := (\hat{M}_{t_2}(d)_{a,b} - \hat{M}_{t_1}(d)_{a,b}) - (M_{t_2}(d)_{a,b} - M_{t_1}(d)_{a,b})$. By construction, 

$$\hat{M}_{t_j}(d)_{a,b} = \text{round}(M_{t_j}(d)_{a,b} + \eta_j)$$ 

\noindent for independently drawn $\eta_j \sim Lap(\epsilon^{-1}T)$ for $j \in \{1,2\}.$ Then $D_{1,2}= \text{round}(\eta_2) - \text{round}(\eta_1)$. As $\eta_1$ and $\eta_2$ are identically distributed, so too are $\text{round}(\eta_1)$ and $\text{round}(\eta_2)$. So then $(\text{round}(\eta_2) - \text{round}(\eta_1))$ and $(\text{round}(\eta_1) - \text{round}(\eta_2))$ are also identically distributed. Thus, 

$$\P(|D_{1,2}| > \alpha) = \P(|\text{round}(\eta_2) - \text{round}(\eta_1)|> \alpha) = 2\P(\text{round}(\eta_2) - \text{round}(\eta_1)> \alpha) = 2\P(\eta_2 - \eta_1 \ge \alpha +1)$$ 

\noindent Replacing $\alpha$ with $(\alpha+1)$ from Lemma \ref{4-lem-difference_laplaces} completes the proof.
\end{proof}

These privacy-accuracy tradeoff theorems described above can be leveraged to \textit{proactively} enable policymakers to set $\epsilon$ based on context-specific goals. In particular, we derive \textit{exact formulas} to set $\epsilon$ in   Theorems \ref{4-thm-beta-nonsuppression_static_accuracy} and \ref{4-thm-beta-nonsuppression_dynamic_accuracy}. To elucidate discussion, suppose policymakers can tolerate some maximal error $\alpha$ with confidence $100(1-\beta)\%$ and still provide effective interventions. Then we can set $\beta \ge \exp(-\epsilon T^{-1}(\alpha + 0.5))$. Solving for $\epsilon$ yields $\epsilon \ge -T(\alpha + 0.5)^{-1}\ln(\beta)$. This yields the following result.

\begin{thm}
\label{4-thm-beta-nonsuppression_static_accuracy}
Any $\beta \in [0,1)$ upper-bounds the result from Theorem \ref{4-thm-nonsuppression_static_accuracy} $\iff \epsilon \ge -T(\alpha + 0.5)^{-1}\ln(\beta)$.
\end{thm}

Alternatively, if policymakers are instead interested in ensuring that trends across different O-D matrices are $\alpha$-inaccurate with $100(1-\beta)\%$ confidence, then we can use Theorem \ref{4-thm-nonsuppression_dynamic_accuracy} to derive feasible values of $\epsilon$ to meet this goal using the Lambert-$W$ function\footnote{The Lambert-$W$ function is a countable family of functions $\{W_{k}: \mathbb{C} \rightarrow \mathbb{C} \text{ } | \text{ } k \in \Z\}$ that satisfies the product-exponential equation $x = W_k(x)\exp(W_k(x))$. That is, for $x,y \in \mathbb{C}$, $x=y\exp(y) \iff y = W_k(x)$ for some $k \in \Z$. In the language of complex analysis, each $W_k$ is a branch of the Lambert-$W$ function. When $x$ is restricted to the real line instead of the complex plane, the domain of the Lambert-$W$ function is $x \ge -\exp(-1)$. When $x < 0$ as well, it is sufficient to consider the principal branch $W_0$ and the lower branch $W_{-1}$ to solve for $y \in \R$. The principal branch $W_0$ corresponds to all $x$ such that $W_0(x) \ge -1$, while the lower branch $W_{-1}$ corresponds to all $x$ such that $W_{-1}(x) \le -1$. See \cite{chatzigeorgiou2013bounds} for a further mathematical details.} \cite{chatzigeorgiou2013bounds}.

\begin{thm}
\label{4-thm-beta-nonsuppression_dynamic_accuracy}
For $\beta \in [0,0.5\exp(-1)]$, $\beta$ upper-bounds the result from Theorem 5 $\iff$
$
\epsilon \ge T(\alpha + 1)^{-1}(-2 - W_{-1}(-2\beta \exp(-2)))
$
where $W_{-1}$ is the lower branch of the Lambert-$W$ function.
\end{thm}
 
\begin{proof}
Let $A = -\epsilon T^{-1}(\alpha+1) - 2$. Then, 
$$\beta \ge \exp(-\epsilon T^{-1}(\alpha+1)) \frac{\epsilon T^{-1}(\alpha+1) + 2}{2} \iff -2\beta\exp(-2) \le A\exp(A)$$ 

Since $\beta \le 0.5\exp(-1)$, we have $-2\beta \exp(-2)) \ge -\exp(-1)$, so $-2\beta \exp(-2))$ is within the domain of the Lambert-$W$ function. Since $-2\beta \exp(-2)) < 0$ there are two potential solutions for $A$ over the reals: one involving the principal branch $W_{0}$ and a second involving the lower branch $W_{-1}$. 

For the first case, we will show that there is no feasible positive $\epsilon$ that works. Since $W_0$ is an increasing function, $-2\beta\exp(-2) \le A\exp(A) \iff W_0(-2\beta\exp(-2)) \le A \iff \epsilon \le T(\alpha + 1)^{-1}(-2 - W_{0}(-2\beta \exp(-2)))$. But this bound is always negative, as $-2\beta \exp(-2) < 0$ implies $W_{0}(-2\beta \exp(-2)) < 0$. So there is no feasible positive $\epsilon$ from the use of the principal branch. 

For the second case, we will derive the inequality stated in the theorem. As $W_{-1}$ is a decreasing function, $-2\beta\exp(-2) \le A\exp(A) \iff W_{-1}(-2\beta\exp(-2)) \ge A \iff \epsilon \ge T(\alpha + 1)^{-1}(-2 - W_{-1}(-2\beta \exp(-2)))$. Unlike the prior case, this bound is indeed a positive quantity. By definition of the Lambert-$W$ function, $W_{-1}(-2\exp(-2)) = -2$. Since $\beta \le 0.5\exp(-1)<1$, we have $-2\exp(-2) < -2\beta\exp(-2)$. Since $W_{-1}$ is decreasing, we deduce that $-2 = W_{-1}(-2\exp(-2)) \ge W_{-1}(-2\beta \exp(-2))$ which implies that the bound is positive.
\end{proof}

\section*{Membership Inference Attack Analysis}
We implement a membership inference attack (MIA) to assess (1) the extent to which aggregated mobility matrices leak individual information that can be exposed by adversaries, and (2) the extent to which differentially private matrices protect this information. We implement the membership inference attack from Pyrgelis et al. (2017) \cite{pyrgelis2017knock}, which uses machine learning tools to attempt to \textit{predict} whether an individual's information is included in an aggregate mobility trace at time $t$ based on the mobility aggregate itself. The machine learning model is trained on previous mobility traces at time $t^* < t$ --- both traces that do and do not include data from the target individual --- where the adversary has knowledge of the ``label'' for the ML model (whether or not the target subscriber's information is included in the aggregate). 

Our version of the membership inference attack corresponds to attack version 2A from Pyrgelis et al. (2017), where the adversary has knowledge of historic participation in past groups of a set of individuals, including the target individual. We believe this version of the attack most closely corresponds to our setting: in our setting, the adversary may have knowledge about a target subscriber's previous movements (between the same set of regions) if they have access to historical mobile phone metadata at the individual level. 

We implement the attack as follows, using our Afghanistan 2020 dataset, which includes 305 days with mobility data. Our other two datasets --- which are each one week long --- are prohibitively short to provide sufficient observations for training a machine learning model for the attack. We focus on the Afghanistan data at the province level, to keep the experiments at a reasonable dimensionality (at the province level there are 34 x 34 = 1,156 OD matrix counts to use as input features to the machine learning model; at the district level there would be 421 x 421 = 177,241 features to use). 

\begin{enumerate}
\item Divide the time series (305 days) into disjoint training sets (the first half of the time period, 152 days) and test sets (the second half of the time period, 153 days).

\item Sample 100 subscribers at random from the set of subscribers with trips in the aggregate mobility matrices, restricting to subscribers that make at least 10 trips during the training period and 10 trips during the test period (to ensure that there are at least 10 positive-instance trips to train the ML model on, and at least 10 positive-instance trips to use for evaluation). These are the subscribers we will run the MIA against. The choice of 100 subscribers is similar to the 150 total individuals that Pyrgelis et al. (2017) uses for the attack. 

\item For each of the 100 sampled subscribers, identify whether they are ``in'' or ``out'' of the OD matrix each day of the training and test sets (i.e. whether they make at least one trip that appears in an OD matrix count on the day in question). Thus, for each subscriber, there are 152 observations for which they are either ``in'' or ``out'' of the OD matrix in the training set and 153 observations for which they are either ``in'' or ``out'' of the OD matrix in the test set. This indicator for whether a subscriber is ``in'' or ``out'' of the matrix is the predictive target of the ML model (see \#5 below).

\item Following Pyrgelis et al. (2017) we create a balanced sample in the training and test sets. To produce this balanced sample we downsample negative observations per subscriber: we keep all of the subscriber’s positive observations (instances where they are ``in'' a specific OD matrix count) and sample the same number of negative observations. Subscribers in our sample make trips on between 20 and 150 days over the duration of the entire period (train and test set together), with a mean of 47 and a median of 32.

\item The input features for the machine learning model are all 34 provinces $\times$ 34 provinces = 1,156 OD matrix counts for the day of the observation (either nonprivate or differentially private). For each subscriber separately, we train a machine learning classifier on the training set to predict whether the subscriber is “in” the OD matrix in question from the input features. We then produce predictions for the test set, and calculate the area under the curve (AUC) score on the test set. For machine learning models, we experiment with a logistic regression (with L1 penalty chosen via three fold cross validation) and a random forest (with an ensemble size of 100 and maximum depth chosen via three fold cross validation).

\item We test the MIA in three settings: (1) using nonprivatized mobility matrices as inputs, (2) using privatized mobility matrices as inputs in both the training and test sets, and (3) using nonprivated mobility matrices in training and privatized mobility matrices in testing. The third option is inspired by Pyrgelis et al.'s (2017) finding that while an adversary that trains on raw data is likely to be inhibited by privatized matrices in the test set, an adversary that mimics the perturbation of mobility matrices in the test set by privatizing the training set will be more successful.  

\item We also test MIA with additional time-related information (in addition to OD matrix counts) as inputs to the model. We consider the inclusion of day-of-week fixed effects to allow the adversary to take advantage of weekly cyclical mobility patterns in inference.

\end{enumerate}

Table S1 provides the results of our experiments with the membership inference attack. We find that the ML-based adversarial approach does have some ability to predict whether individual subscribers are included in an OD matrix from the nonprivatized aggregated counts (average AUC = 0.61-0.62 for a logistic regression 0.63-0.64 for a random forest). Adding time fixed effects does not improve the accuracy of the adversary. In comparison to the two mobility datasets from taxi drivers and subway systems studied in Pyrgelis et al. (AUC = 0.81-0.99), the performance of the adversary in our setting is much lower. This may be because the number of individuals underlying our dataset is much larger: our 305-day Afghanistan dataset includes transactions from around 7 million subscribers, whereas the two datasets analyzed in Pyrgelis et al. (2017) include data from 536 and 10,000 individuals. 

In our setting differential privacy provides a small amount of additional privacy protection in the setting where the model is trained on raw (nonprivatized OD matrices). In particular, the random forest in this setting achieves an AUC of 0.623-0.624 in comparison to an AUC of 0.631 for a model evaluated on nonprivatized data. Consistent with the results in Pyrgelis et al. (2017), we find that these privacy gains are attenuated when the model is trained on privatized matrices and can learn to mimic the pattern of noise injected (AUC = 0.633). Across the board, we see no difference between attacks using privated and nonprivatized matrices when the logistic regression is used for the attack (although the attacker's accuracy, even with nonprivatized matrices, is lower for the logistic regression than for the random forest).

\section*{Supplemental Figures and Tables}

\renewcommand{\thefigure}{\textbf{S1}} 
\begin{figure}[!h]
\centering
\includegraphics[width=\textwidth]{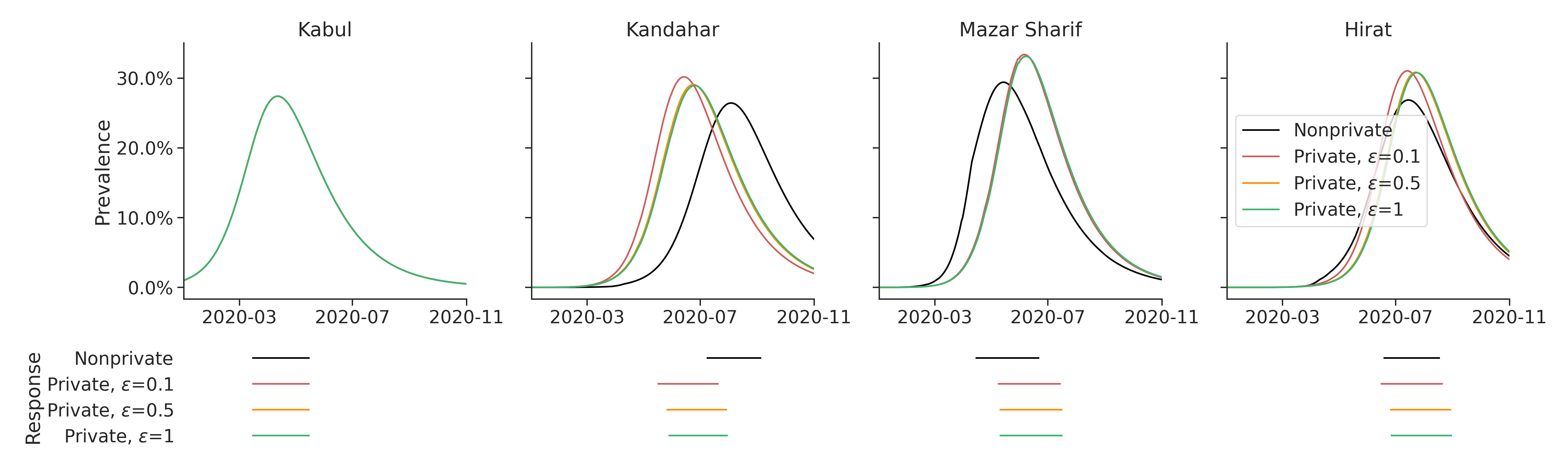}
\includegraphics[width=\textwidth]{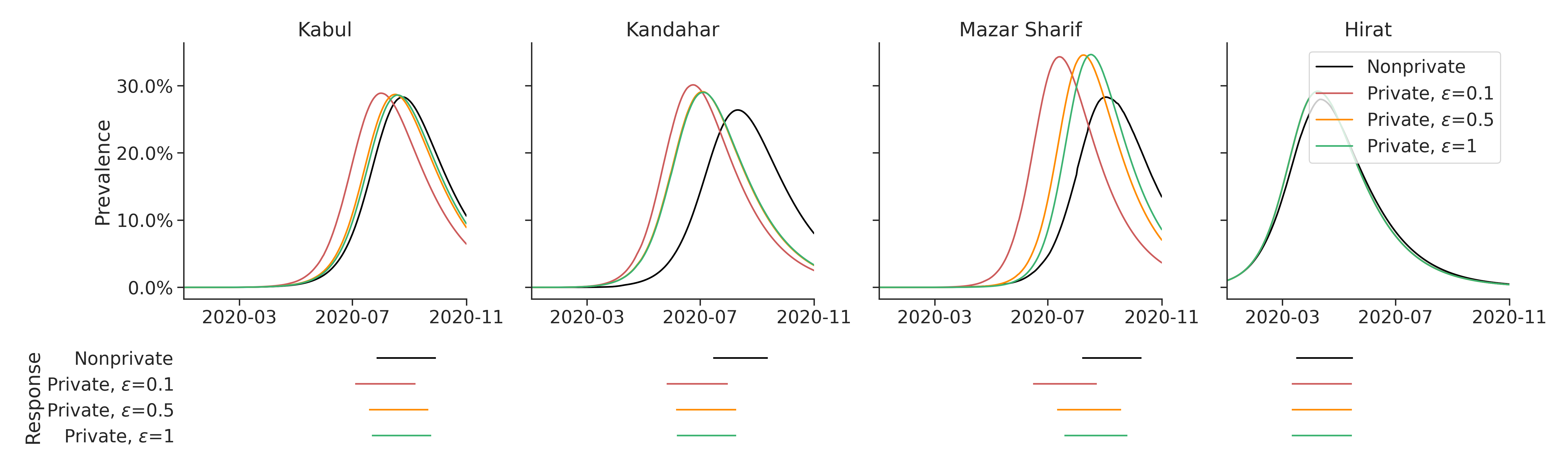}
\includegraphics[width=\textwidth]{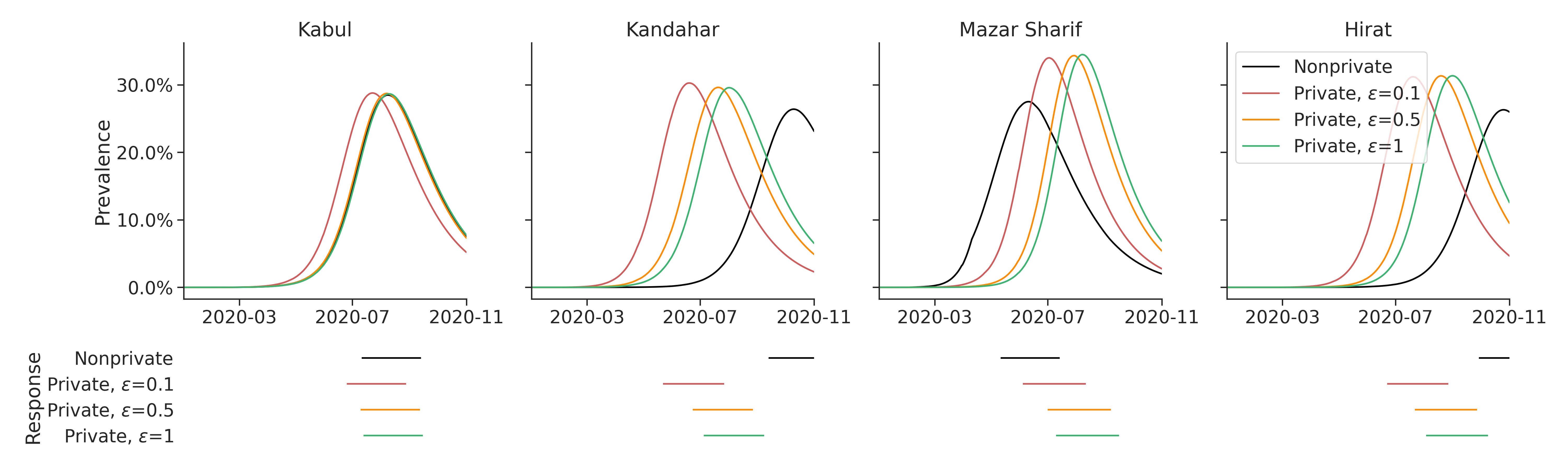}
\caption{Epidemic curves based on mobility-based SIR models and O-D matrices derived from call detail records. Top: Pandemic initiating in Kabul. Middle: Pandemic initiating in Hirat. Bottom: Pandemic initiating randomly.}
\label{fig:epicurves}
\end{figure}

\renewcommand{\thefigure}{\textbf{S2}}  
\begin{figure}[!h]
\centering
\includegraphics[width=\textwidth]{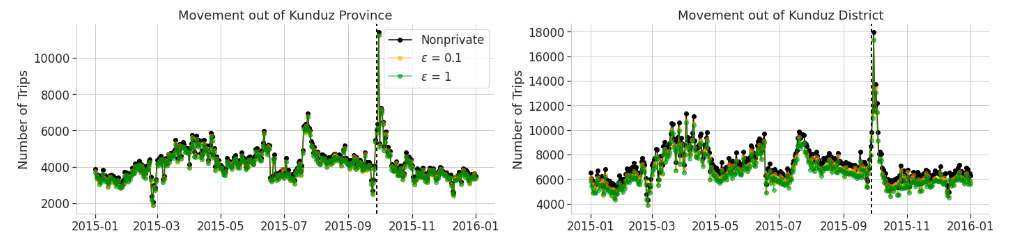}
\includegraphics[width=\textwidth]{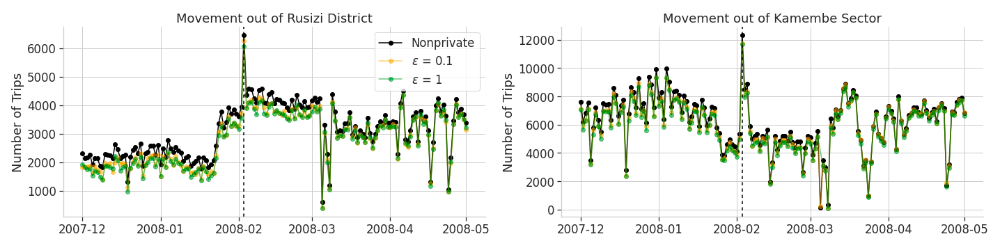}
\caption{Total out-migration from areas used in simulations of humanitarian response to natural disasters and violent events, calculated from call detail records. Top: Out-migration during the Battle of Kunduz in Afghanistan. Bottom: Out-migration following the Lake Kivu Earthquake in Rwanda.}
\label{fig:outmigration}
\end{figure}

\renewcommand{\thefigure}{\textbf{S3}}  
\begin{figure}[!h]
\centering
\includegraphics[width=\textwidth]{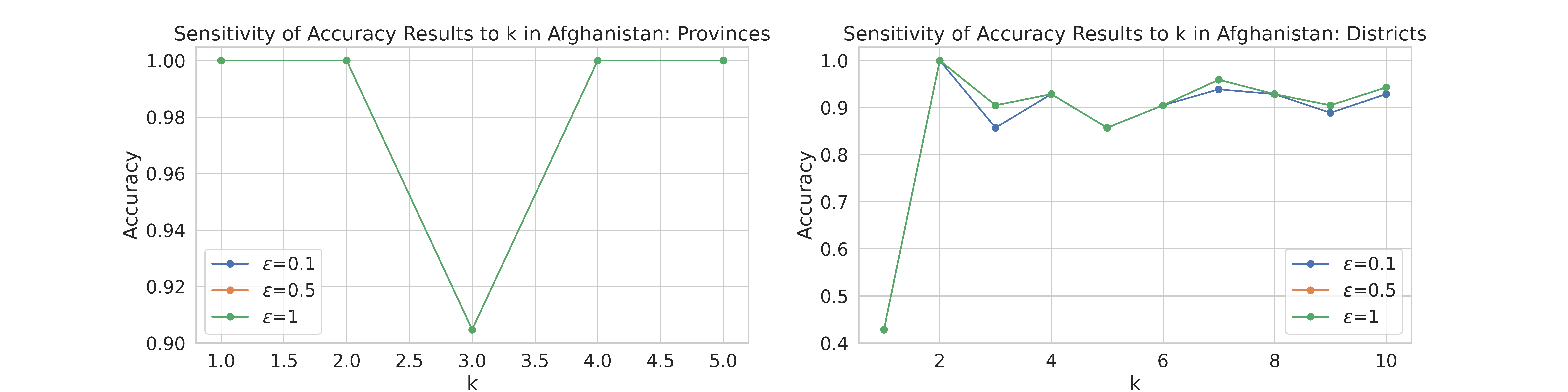}
\includegraphics[width=\textwidth]{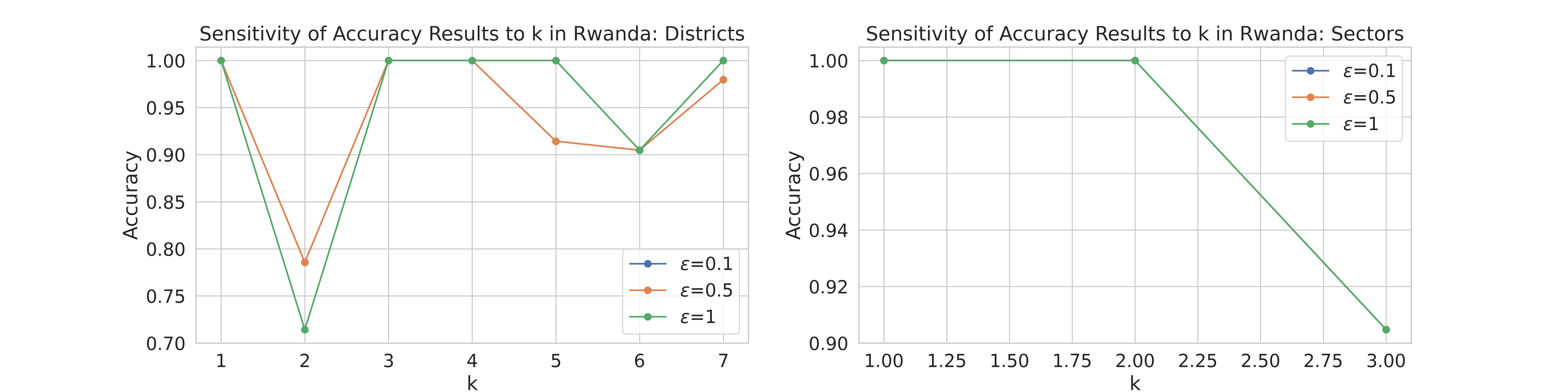}
\caption{Sensitivity of results on accuracy for identifying top-$k$ regions of out migration in simulations of natural disasters and violent events to value of $k$. Left: admin-2 level. Right: admin-3 level. Lines are not present where there are fewer than $k$ regions of out-migration after suppression of small counts.}
\label{fig:ksensitivity}
\end{figure}

\renewcommand{\thetable}{\textbf{S1}}
\begin{table}[h]
\centering
\begin{tabular}{@{}lllll@{}}
\toprule
\multirow{2}{*}{}      & \multicolumn{2}{l}{\textbf{Features: OD counts only}}                   & \multicolumn{2}{l}{\textbf{Features: OD counts and weekday dummies}} \\
                       & \textit{\textbf{Logistic regression}} & \textit{\textbf{Random forest}} & \textit{\textbf{Logistic regression}}  & \textit{\textbf{Random forest}} \\ \midrule
\multicolumn{5}{l}{\textit{\textbf{Panel A: Train on non-private data}}}                                                                                                    \\
\textbf{Non-private}   & 0.619                                 & 0.643                           & 0.619                                  & 0.631                           \\
\textbf{$\epsilon$ = 0.1} & 0.619                                 & 0.628                           & 0.619                                  & 0.623                           \\
\textbf{$\epsilon$ = 0.5} & 0.620                                 & 0.626                           & 0.620                                  & 0.624                           \\
\textbf{$\epsilon$ = 1}   & 0.620                                 & 0.626                           & 0.620                                  & 0.624                           \\ \\
\multicolumn{5}{l}{\textit{\textbf{Panel B: Train on private data}}}                                                                                                       \\
\textbf{Non-private}   & 0.619                                 & 0.643                           & 0.619                                  & 0.631                           \\
\textbf{$\epsilon$ = 0.1} & 0.614                                 & 0.636                           & 0.614                                  & 0.633                           \\
\textbf{$\epsilon$ = 0.5} & 0.614                                 & 0.634                           & 0.614                                  & 0.633                           \\
\textbf{$\epsilon$ = 1}   & 0.614                                 & 0.634                           & 0.614                                  & 0.633                           \\ \bottomrule
\end{tabular}
\label{table:mia}
\caption{Results of membership inference attack on raw OD matrices vs. privatized matrices (with three different values of $\epsilon$ using trip-level privacy). The attack is run on our 305-day Afghanistan dataset. The first panel trains the attack model on raw OD matrices; the second panel trains the attack model on privatized matrices. Each table cell represents the average AUC score for the attack in the test set on 100 sampled subscribers.}
\end{table}

\renewcommand{\thetable}{\textbf{S2}}

\begin{table}[h]
\centering

\begin{tabular}{@{}lcccccc@{}}
\toprule
\multirow{2}{*}{}                & \multicolumn{2}{l}{\textbf{Rwanda 2008}}                                      & \multicolumn{2}{l}{\textbf{Afghanistan 2015}}                                  & \multicolumn{2}{l}{\textbf{Afghanistan 2020}}                                   \\
                                 & \multicolumn{1}{r}{\textbf{Districts}} & \multicolumn{1}{r}{\textbf{Sectors}} & \multicolumn{1}{r}{\textbf{Provinces}} & \multicolumn{1}{r}{\textbf{Districts}} & \multicolumn{1}{r}{\textbf{Provinces}} & \multicolumn{1}{r}{\textbf{Districts}} \\ \midrule
\textbf{Non-private}              & 49.3\%                                 & 78.1\%                               & 75.9\%                                 & 96.7\%                                & 73.6\%                                 & 97.8\%                                 \\
\multicolumn{1}{r}{\textbf{$\epsilon$ = 0.1}} & 47.6\%                                 & 88.3\%                               & 65.9\%                                 & 88.2\%                                & 64.1\%                                 & 88.2\%                                 \\
\multicolumn{1}{r}{\textbf{$\epsilon$ = 0.5}} & 49.2\%                                 & 90.6\%                               & 75.7\%                                 & 97.1\%                                & 73.5\%                                 & 97.7\%                                 \\
\multicolumn{1}{r}{\textbf{$\epsilon$ = 1.0}} & 49.4\%                                 & 78.0\%                               & 76.0\%                                 & 96.5\%                                & 73.6\%                                 & 97.5\%                                 \\ \bottomrule
\end{tabular}
\label{table:suppression}
\caption{Share of OD matrix counts that have fewer than 15 trips in each dataset. For privatized data, these regions are suppressed in our empirical simulations --- that is, their OD matrix counts are set to 0. The large number of OD matrix counts with fewer than 15 trips across datasets reflects the sparsity of the OD matrices: many region-to-region pairs have little direct travel, particularly those that are far apart from one another.}
\end{table}

\renewcommand{\thetable}{\textbf{S3}}

\begin{table}[h]
\centering
\small

\begin{tabular}{@{}lllllllll@{}}
\toprule
                                     & \multicolumn{4}{l}{\textbf{Admin-2 level}}                            & \multicolumn{4}{l}{\textbf{Admin-3 level}}                            \\ \midrule
                                     & Non-Private & \textit{$\epsilon$=0.1} & \textit{$\epsilon$=0.5} & \textit{$\epsilon$=1} & \textit{Non-private} & \textit{$\epsilon$=0.1} & \textit{$\epsilon$=0.5} & \textit{$\epsilon$=1} \\ \midrule
\multicolumn{9}{l}{\textit{Panel A: Battle of Kunduz in Afghanistan, T set at 99th percentile of daily trips}}                                                                                                          \\
Total out-migration                  & 49,994               & 42,301         & 40,982         & 40,920       & 87,007               & 83,471         & 74,755         & 74,507       \\
\% error in total out-migration & 0.00\%               & 15.39\%        & 18.03\%        & 18.15\%      & 0.00\%               & 4.06\%         & 14.08\%        & 14.37\%      \\
Acc. of top-$k$ regions ($k$ = 3)    & 100.00\%             & 85.71\%        & 85.71\%        & 85.71\%      & 100.00\%             & 85.71\%        & 85.71\%        & 85.71\%      \\ \\
\multicolumn{9}{l}{\textit{Panel B: Battle of Kunduz in Afghanistan, T set at 95th percentile of daily trips}}                                                                                                          \\
Total out-migration                  & 49,994               & 41,492         & 40,173         & 40,111       & 87,007               & 79,995         & 71,298         & 71,035       \\
\% error in total out-migration & 0.00\%               & 17.01\%        & 19.64\%        & 19.77\%      & 0.00\%               & 8.06\%         & 18.05\%        & 18.36\%      \\
Acc. of top-$k$ regions ($k$ = 3)    & 100.00\%             & 85.71\%        & 85.71\%        & 85.71\%      & 100.00\%             & 85.71\%        & 85.71\%        & 85.71\%      \\  \\
\multicolumn{9}{l}{\textit{Panel C: Battle of Kunduz in Afghanistan, T set at 90th percentile of daily trips}}                                                                                                          \\
Total out-migration                  & 49,994               & 40,389         & 39,070         & 39,008       & 87,007               & 75,157         & 66,374         & 66,144       \\
\% error in total out-migration & 0.00\%               & 19.21\%        & 21.85\%        & 21.97\%      & 0.00\%               & 13.62\%        & 23.71\%        & 23.98\%      \\
Acc. of top-$k$ regions ($k$ = 3)    & 100.00\%             & 85.71\%        & 85.71\%        & 85.71\%      & 100.00\%             & 85.71\%        & 85.71\%        & 85.71\%      \\  \\
\multicolumn{9}{l}{\textit{Panel D: Lake Kivu Earthquake in Rwanda, T set at 99th percentile of daily trips}}                                                                                                           \\
Total out-migration                  & 32,627               & 30,710         & 30,752         & 30,740       & 51,102               & 57,657         & 47,504         & 47,259       \\
\% error in total out-migration & 0.00\%               & 5.00\%         & 4.87\%         & 4.91\%       & 0.00\%               & 12.83\%        & 9.00\%         & 7.52\%       \\
Acc. of top-$k$ regions ($k$ = 3)    & 100.00\%             & 100.00\%       & 100.00\%       & 100.00\%     & 100.00\%             & 80.95\%        & 85.71\%        & 85.71\%      \\  \\
\multicolumn{9}{l}{\textit{Panel E: Lake Kivu Earthquake in Rwanda, T set at 95th percentile of daily trips}}                                                                                                           \\
Total out-migration                  & 32,627               & 29,823         & 29,880         & 29,883       & 51,102               & 56,365         & 46,231         & 45,992       \\
\% error in total out-migration & 0.00\%               & 8.59\%         & 8.42\%         & 8.41\%       & 0.00\%               & 10.30\%        & 9.53\%         & 10.00\%      \\
Acc. of top-$k$ regions ($k$ = 3)    & 100.00\%             & 100.00\%       & 100.00\%       & 100.00\%     & 100.00\%             & 85.71\%        & 85.71\%        & 85.71\%      \\  \\
\multicolumn{9}{l}{\textit{Panel E: Lake Kivu Earthquake in Rwanda, T set at 90th percentile of daily trips}}                                                                                                           \\
Total out-migration                  & 32,627               & 28,819         & 28,866         & 28,868       & 51,102               & 55,273         & 45,169         & 44,935       \\
\% error in total out-migration & 0.00\%               & 11.67\%        & 11.53\%        & 11.52\%      & 0.00\%               & 8.16\%         & 11.61\%        & 12.07\%      \\
Acc. of top-$k$ regions ($k$ = 3)    & 100.00\%             & 100.00\%       & 100.00\%       & 100.00\%     & 100.00\%             & 85.71\%        & 85.71\%        & 85.71\%      \\ \bottomrule
\end{tabular}

\label{table:geoappendix}
\caption{Replication of Table 3 with different values of $T$. Recall that $T$ represents a cut-off where, for any subscriber making more than $T$ trips on a given day, only T of those trips are sampled for inclusion in the calculation of the OD matrix. We test three values of $T$ for both the Battle of Kunduz in Afghanistan and the Lake Kivu earthquake in Rwanda. We take the 99th, 95th, and 90th percentiles of daily trips that subscribers make in each setting. For Afghanistan, this corresponds to $T$ = 6, 4, and 3 at the admin-2 level and $T$ = 10, 6, and 4 at the admin-3 level. For Rwanda, these percentiles correspond to $T$ = 14, 8, and 6 at the admin-2 level, and $T$ = 20, 11, and 9 at the admin-3 level.}
\end{table}

\end{document}